	\documentclass[
		paper=letter,10pt,american,pagesize,%
		twocolumn,headings=small,title=small,DIV=15,
		abstract=true,
		headinclude=true,footinclude=true,footlines=-5,twoside=semi%
	]{scrartcl}
	\def\docclass{my-siam}
	\def\version{arxiv}

	\def\draftmode{false} 

\usepackage[utf8]{inputenc}
\makeatletter
\usepackage{xifthen}

\newcommand\iflipics[2]{\ifthenelse{\equal{\docclass}{lipics}}{#1}{#2}}
\newcommand\ifkoma[2]{\ifthenelse{\equal{\docclass}{koma}}{#1}{#2}}
\newcommand\ifieee[2]{\ifthenelse{\equal{\docclass}{ieee}}{#1}{#2}}
\newcommand\ifsiam[2]{\ifthenelse{\equal{\docclass}{siam}}{#1}{#2}}
\newcommand\ifmysiam[2]{\ifthenelse{\equal{\docclass}{my-siam}}{#1}{#2}}
\ifthenelse{ \equal{\docclass}{lipics} \OR \equal{\docclass}{koma} \OR \equal{\docclass}{ieee} \OR \equal{\docclass}{siam} \OR \equal{\docclass}{my-siam} }{
	\PackageInfo{paper}{Building paper with docclass = \docclass} 
}{
	\PackageWarning{paper}{docclass = "\docclass", but must be one of "lipics", "koma", "ieee", "siam", "my-siam"}
}

\newcommand\ifmanuscript[2]{\ifthenelse{\equal{\version}{manuscript}}{#1}{#2}}
\newcommand\ifarxiv[2]{\ifthenelse{\equal{\version}{arxiv}}{#1}{#2}}
\newcommand\ifsubmission[2]{\ifthenelse{\equal{\version}{submission}}{#1}{#2}}
\newcommand\ifproceedings[2]{\ifthenelse{\equal{\version}{proceedings}}{#1}{#2}}
\ifthenelse{ 
	\equal{\version}{manuscript} 
	\OR \equal{\version}{arxiv} 
	\OR \equal{\version}{submission} 
	\OR \equal{\version}{proceedings} 
}{
	\PackageInfo{paper}{Building paper version = \version} 
}{
	\PackageWarning{paper}{version = "\version", but must be one of "manuscript", "arxiv", "submission", "proceedings"}
}

\newcommand\ifdraft[2]{\ifthenelse{\equal{\draftmode}{true}}{#1}{#2}}
\ifthenelse{ \equal{\draftmode}{true} \OR \equal{\draftmode}{false} }{
	\PackageInfo{paper}{Building paper with draftmode = \draftmode} 
}{
	\PackageWarning{paper}{draftmode = "\draftmode", but must be "true" or "false"}
}

\usepackage[T1]{fontenc}
\ifieee{}{
	\usepackage{lmodern}
	\usepackage{slantsc}
}

\usepackage{babel}
\input{ushyphex.tex} 

\usepackage{array,multicol}
\ifieee{
	\usepackage[cmex10]{amsmath,mathtools}
	\usepackage{amsfonts,amssymb}
}{
	\usepackage{amsmath,amsfonts,amssymb,mathtools}
}

\usepackage{mleftright}\mleftright 
\usepackage{relsize,xspace,booktabs,adjustbox,needspace,pbox,relsize}
\ifieee{
	
	\usepackage{enumitem}
}{
	\usepackage{enumitem}
}
\usepackage{graphicx}

\usepackage{colonequals}

\usepackage{wref}
\usepackage[bibtex]{url-doi-arxiv}

\ifsiam{
	\usepackage{ltexpprt}
}{}

\newdimen\makeboxdimen

\ifthenelse{\equal{\docclass}{koma} \OR \equal{\docclass}{my-siam}}{
	\setlength\parindent{1.5em}
	\usepackage[headsepline]{scrlayer-scrpage}
	\pagestyle{scrheadings}
	\clearscrheadfoot
	\AtBeginDocument{%
		\automark[section]{}%
	}
	\ohead{\pagemark}
	\rehead{\mytitle}
	\lohead{\headmark}
	\addtokomafont{caption}{\sffamily\small}
	\addtokomafont{captionlabel}{\sffamily\textbf}
	\setcapmargin{2em}
}{}
\ifthenelse{\equal{\docclass}{my-siam}}{
	\setcapmargin{1em}
	\setcapindent{0em}
}{}
\ifmysiam{
	\setlength\parskip{0pt}
	\RedeclareSectionCommand[
		beforeskip=-1.25\baselineskip,
		afterskip=0.75\baselineskip,
	]{section}
	\RedeclareSectionCommand[
		beforeskip=-1\baselineskip,
		afterskip=-1.5em,
	]{subsection}
	\RedeclareSectionCommand[
		beforeskip=-1\baselineskip,
		afterskip=-1.5em,
	]{subsubsection}
	\RedeclareSectionCommand[
		beforeskip=-.25\baselineskip,
		indent=1.5em,
		afterskip=-1em,
	]{paragraph}
}{}

\AtBeginDocument{%
	\let\mytitle\@title%
}

\newcommand\shorttitle[1]{%
	\iflipics{%
		\titlerunning{#1}%
	}{}%
	\AtBeginDocument{%
		\def\mytitle{#1}%
	}%
}

\let\oldthebibliography\thebibliography
\renewcommand\thebibliography[1]{%
	\oldthebibliography{#1}%
	\pdfbookmark[1]{References}{}%
}

\usepackage{lscape} 

\ifkoma{
	\usepackage{float}
	\floatstyle{plain}
	\usepackage{newfloat}
	\DeclareFloatingEnvironment[%
			name=Algorithm,%
			placement=thb,%
		]{algorithm}
}{}
\iflipics{
	\usepackage{newfloat}
	\DeclareFloatingEnvironment[%
			name=Algorithm,%
			placement=thb,%
		]{algorithm}
}{}

\ifmysiam{
	\setcounter{topnumber}{3}
	\setcounter{bottomnumber}{3}
	\setcounter{totalnumber}{3}     
	\setcounter{dbltopnumber}{3}    

}{}
\ifsiam{
	\setcounter{topnumber}{3}
	\setcounter{bottomnumber}{3}
	\setcounter{totalnumber}{3}     
	\setcounter{dbltopnumber}{3}    

}{}

\usepackage{dcolumn}

\usepackage{textcomp} 
\usepackage{listings}

\usepackage{tikz}

\usetikzlibrary{positioning,arrows.meta}
\usetikzlibrary{backgrounds,calc,trees,graphs}

\pgfdeclarelayer{background}
\pgfsetlayers{background,main}

\usepackage{pgfplots}

\usetikzlibrary{external}
\tikzsetexternalprefix{plots/externalized/}
\tikzset{external/export=false} 

\newcommand{\externalizedpicture}[1]{%
	\tikzset{external/export=true}%
	\tikzsetnextfilename{#1}%
	\ignorespaces%
}

\iflipics{
	\newtheorem{fact}[theorem]{Fact}
	\newcommand\thmemph[1]{\emph{#1}}
	\newenvironment{proofof}[1]{%
		\begin{proof}[{{Proof of #1{}}}]%
	}{%
		\end{proof}%
	}
}{}
\newcommand\cleartheorem[1]{%
	\expandafter\let\csname#1\endcsname\relax
	\expandafter\let\csname c@#1\endcsname\relax
	\newcounter{#1} 
}
\ifsiam{
	\newenvironment{proofof}[1]{%
			\begin{proof}[{{#1{}}}]%
		}{%
			\end{proof}%
		}
	\newcommand\thmemph[1]{\emph{#1}}
	
	\cleartheorem{conjecture}
	\newtheorem{conjecture}[@theorem]{Conjecture}
	\cleartheorem{remark}
	\newtheorem{remark}[@theorem]{Remark}
	\cleartheorem{lemma}
	\newtheorem{lemma}[@theorem]{Lemma}
	\cleartheorem{fact}
	\newtheorem{fact}[@theorem]{Fact}
	\cleartheorem{corollary}
	\newtheorem{corollary}[@theorem]{Corollary}
	\cleartheorem{axiom}
	
	\cleartheorem{property}
	\newtheorem{property}[@theorem]{Property}
	\cleartheorem{proposition}
	\newtheorem{proposition}[@theorem]{Proposition}
	
}{}
\ifthenelse{\equal{\docclass}{lipics} \OR \equal{\docclass}{siam}}{}{
	\usepackage[amsmath,hyperref,thmmarks]{ntheorem}
	
	\theorembodyfont{\slshape}
	\theoremseparator{:}
	\newtheoremstyle{proofstyle}%
	  {\item[\theorem@headerfont\hskip\labelsep ##1\theorem@separator]}%
	  {\item[\theorem@headerfont\hskip\labelsep ##3\theorem@separator]}
	
	\theorempreskip{\topsep} 

	\theoremsymbol{\adjustbox{scale=.8}{$\triangleleft\mkern-1mu$}}
	
	\newtheorem{theorem}{Theorem}[section]
	
	\theoremstyle{plain}
	\theorempreskip{\topsep}
	
	\newtheorem{proposition}[theorem]{Proposition}
	\newtheorem{lemma}[theorem]{Lemma}
	\newtheorem{conjecture}[theorem]{Conjecture}
	\newtheorem{corollary}[theorem]{Corollary}

	\theoremstyle{plain}
	\theorembodyfont{\upshape}

	\newtheorem{remark}[theorem]{Remark}

	\theoremsymbol{\raisebox{-.25ex}{$\Box$}}
	\qedsymbol{\raisebox{-.25ex}{$\Box$}}
	
	\theoremstyle{proofstyle}
	\newtheorem{proof}{Proof}
	\newenvironment{proofof}[1]{%
		\begin{proof}[{{Proof of #1{}}}]%
	}{%
		\end{proof}%
	}
	
	\newcommand\thmemph[1]{\textit{#1}}
}

\iflipics{
	\newenvironment{thmenumerate}[2][]{%
		\begin{enumerate}[
			label={\textsf{\textbf{\color{darkgray}{\makebox[\widthof{(a)}][c]{\textup{(\alph*)}}}}}},
			ref={\ref{#2}\kern.1em--\kern.1em(\alph*)},
			itemsep=0pt,
			topsep=.5ex,
			leftmargin=1.75em,
			#1
		]%
	}{%
		\end{enumerate}%
	}
}{
	
}

\newcommand*\ie{\mbox{i.\hspace{.2ex}e.}}

\newcommand*\wrt{\mbox{w.\hspace{.2ex}r.\hspace{.2ex}t.}\xspace}
\newcommand*\aka{\mbox{a.\hspace{.2ex}k.\hspace{.2ex}a.}\xspace}

\newcommand\R{\mathbb R}
\newcommand\N{\mathbb N}

\usepackage{fixmath}

\newcommand{\ESymbol}{\mathbb{E}}

\newcommand{\ProbSymbol}{\ensuremath{\mathbb{P}}}

\DeclarePairedDelimiterXPP\Prob[1]{\ProbSymbol}[]{}{%
	#1%
}
\DeclarePairedDelimiterXPP\E[1]{\ESymbol}[]{}{%
	#1%
}
\DeclarePairedDelimiterXPP\Eover[2]{\ESymbol_{#1}}[]{}{%
	#2%
}
\DeclarePairedDelimiterXPP\ProbIn[2]{\ProbSymbol_{#1}}[]{}{%
	#2%
}
\providecommand{\Prob}{} 
\providecommand{\ProbIn}{} 
\providecommand{\E}{} 
\providecommand{\Eover}{} 

\newcommand{\surroundedmath}[3]{
	\mathchoice{
		#1{#2{#3}#2}%
	}{
		#1{#3}%
	}{
		#1{#3}%
	}{
		#1{#3}%
	}%
}
\newcommand\rel[1]{\surroundedmath{\mathrel}{\:}{#1}}

\iflipics{}{
	\makeatletter
	\let\oldalign\align
	\let\endoldalign\endalign
	\renewenvironment{align}{%
		\begingroup%
		\let\oldhalign\halign
		\def\halign{%
			\let\oldbreak\\%
			\def\nonnumberbreak{\nonumber\oldbreak*}%
			\def\\{%
				\@ifstar{\nonnumberbreak}{\oldbreak}%
			}%
			\oldhalign%
		}
		\oldalign%
	}{%
		\endoldalign%
		\endgroup%
	}
}
\newcommand*\numberthis[1][]{\stepcounter{equation}\tag{\theequation}}

\allowdisplaybreaks[3]

\newcommand\splitaftercomma[1]{%
  \begingroup
  \begingroup\lccode`~=`, \lowercase{\endgroup
    \edef~{\mathchar\the\mathcode`, \penalty0 \noexpand\hspace{0pt plus .25em}}%
  }\mathcode`,="8000 #1%
  \endgroup
}

\usepackage{soul} 

\def\mydots{\xleaders\hbox to.5em{\hfill.\hfill}\hfill}
\newlength\tmpLenNotations

\ifdraft{%
    \usepackage[switch]{lineno} 
    \linenumbers

	\overfullrule=6mm
	
	\usepackage[color,notref,notcite]{showkeys}
	\definecolor{refkey}{gray}{.99}
	\colorlet{labelkey}{green!60!black!60}
	
	\usepackage[inline,nolabel]{showlabels}

	\showlabels{cite}
	\showlabels{citealt}
	\showlabels{citealp}
	\showlabels{citet}
	\showlabels{Citet}
	\showlabels{citep}
	\showlabels{citeauthor}
	\showlabels{Citeauthor}
	\showlabels{citefullauthor}
	\showlabels{citeyear}
	\showlabels{citeyearpar}
	\showlabels{wref}
	\showlabels{wpref}
	\showlabels{wtpref}
	\showlabels{wildref}
	\showlabels{wildpageref}
	\showlabels{wildtpageref}
}{}

\iflipics{
	\ifmanuscript{\hideLIPIcs}{}
	\ifarxiv{\hideLIPIcs}{}
	\ifsubmission{}{\nolinenumbers}
}{}

\ifdraft{}{%
	\usepackage{microtype}
}

\hypersetup{
	final,
	unicode=true, 
	bookmarks=true,
	bookmarksnumbered=true,
	bookmarksdepth=2,
	bookmarksopen=true,
	breaklinks=true,
	hidelinks,
}

\newsavebox\tmpbox

\iflipics{
	\let\oldparagraph\paragraph
	\renewcommand\paragraph[1]{%
		\oldparagraph*{#1}%
	}
}{
	\let\oldparagraph\paragraph
	\renewcommand\paragraph[1]{%
		\oldparagraph{#1.}%
	}
}
\ifmysiam{
	\let\oldsubsection\subsection
	\renewcommand\subsection[1]{%
		\oldsubsection{#1.}%
	}
	\let\oldsubsubsection\subsubsection
	\renewcommand\subsubsection[1]{%
		\oldsubsubsection{#1.}%
	}
}{}
\ifsiam{
	\let\oldsubsection\subsection
	\renewcommand\subsection[1]{%
		\oldsubsection{#1.}%
	}
	\let\oldsubsubsection\subsubsection
	\renewcommand\subsubsection[1]{%
		\oldsubsubsection{#1.}%
	}
}{}

\let\epsilon\varepsilon

\def\myacknowledgements{}
\ifkoma{
	
}{}
\ifieee{
	
}{}
\ifsiam{
	
}{}
\ifmysiam{
	
}{}

\makeatother

\usepackage{hyperref}

\pgfplotsset{
    tick label style={font=\small},
    label style={font=\normalsize},
}

\title{%
    Towards the 5/6-Density Conjecture\\
    of Pinwheel Scheduling%
    \ifproceedings{
        \thanks{
            An extended version including proofs is available online: \url{https:/arxiv.org/xxxxxxxx}
        }}{}
}
\shorttitle{Towards the 5/6-Density Conjecture in Pinwheel Scheduling} 

	\newcommand\email[1]{\texttt{#1}}
	\ifsubmission{
	}{
		\author{%
			Leszek G\k asieniec%
				\footnote{U. of Liverpool, UK, 
				\email{l.a.gasieniec\,@\,liverpool.ac.uk}}
		\and 
			Benjamin Smith%
				\footnote{U. of Liverpool, UK, 
				\email{b.m.smith\,@\,liverpool.ac.uk}}
		\and
			Sebastian Wild%
				\footnote{U. of Liverpool, UK, 
				\email{sebastian.wild\,@\,liverpool.ac.uk}}
		}
	}
	
	\date{\small\today}
\ifsiam{
	\ifsubmission{\date{}}{}
	
	\fancyfoot[R]{\scriptsize{Copyright \textcopyright\ 2022\\
	Copyright for this paper is retained by authors}}
	
	\ifsubmission{}{\setcounter{page}{1}}
	\let\oldabstract\abstract\let\oldendabstract\endabstract
	\renewenvironment{abstract}{%
		\noindent\begin{minipage}{\linewidth}%
		\small\setlength{\baselineskip}{9pt}\oldabstract%
	}{%
		\oldendabstract\end{minipage}%
	}
}{}

\usepackage{mathtools}
\DeclarePairedDelimiter{\floor}{\lfloor}{\rfloor}
\DeclarePairedDelimiter{\ceil}{\lceil}{\rceil}

\DeclareMathOperator{\len}{length}

\ifproceedings{
	\let\wpref\wref
	\let\wtpref\wref
}{}

\begin{document}

\ifsubmission{\onecolumn}{}

\maketitle

\ifsubmission{\begin{center}\begin{minipage}{.7\linewidth}}{}
\begin{abstract}\ifsubmission{\normalsize}{}
    Pinwheel Scheduling aims to find a perpetual schedule for unit-length tasks on a single machine 
    subject to given maximal time spans (\aka frequencies) between any two consecutive executions of the same task.
    The density of a Pinwheel Scheduling instance is the sum of the inverses of these task frequencies;
    the $5/6$-Conjecture (Chan and Chin, 1993) states that any Pinwheel Scheduling instance with density at most $5/6$ is schedulable.
    We formalize the notion of Pareto surfaces for Pinwheel Scheduling and exploit novel structural insights to engineer an efficient algorithm for computing them.
    This allows us to 
    (1) confirm the $5/6$-Conjecture for all Pinwheel Scheduling instances with at most 12 tasks and (2) to prove that a given list of only 23 schedules solves \emph{all} schedulable Pinwheel Scheduling instances with at most 5 tasks.
\end{abstract}
\ifsubmission{\end{minipage}\end{center}}{}

\ifdraft{%
    {\color{red}
        \noindent\textbf{TO DO FOR ALENEX FINAL VERSION:}
        \begin{itemize}
            \item explain trie-based search, and briefly compare running time with trie based search (ratio at k=11 for motivation, search proportion at some sample k values).
            \item comment on $k$d-trees; can claim high dimensions make other options less desirable/effective
            \item read through for simplifications.
            \item repeat table 2
            \begin{itemize}
                \item extend opt data, and trim out less reliable runs
            \end{itemize}
            \item repeat figure 3
            \begin{itemize}
                \item make figure
                \item add repeats/trim less reliable areas
            \end{itemize}
        \end{itemize}
        \textbf{TO DO LATER (journal/dissertation):}
        \begin{itemize}
            \item upgrades:
            \begin{itemize}
                \item certificates
                \item searching:
                \begin{itemize}
                    \item perfect kd trees (low update count means that we can rebuild instead of updating and thus use a better data structure)
                    \item generalise density based trie search past $a_0$ (to take $a_0$ from 2 to 4, you increase density by 1/2. use the generalisation of this to decrease pointless branching at the search stage.)
                    \item search in parallel with folding
                \end{itemize}
                \item solving:
                \begin{itemize}
                    \item process filler using rounding.
                    \item memoization
                    \item order optimised Foresight (use number of occurrences of a task in the schedule when determining order of task appearance to find shorter schedules and find them quicker).
                \end{itemize}
            \end{itemize}
            \item minimum solution length~-- theory and practice
            \item tight subsets~-- testing some hypotheses
            \item proof~-- density 1 == unfolding
        \end{itemize}
    }
}{}

\ifsubmission{%
	\clearpage
	\addtocounter{page}{-1}
	\twocolumn
}{}

\section{Introduction}

An instance of the Pinwheel Scheduling problem is defined by $k$ positive integer \emph{frequencies} 
$A = (a_1, a_2, \ldots, a_k)$, $a_1\le a_2\le\cdots\le a_k$, 
and is solved by producing a valid \emph{Pinwheel schedule} 
(if one exists, \ie, if the problem is \emph{schedulable}), or by stating that no such schedule exists.
A (Pinwheel) schedule $S_\infty=s_1 s_2\ldots$ is an infinite sequence over $[k] = \{1,\ldots, k\}$; 
it is \emph{valid} (for $A$) if every task $i$ is scheduled at least every \mbox{$a_i$} days.
Formally, any contiguous subsequence $s_t \ldots s_{t+a_i-1}$ of length $a_i$ contains 
at least one occurrence of $i$, for $i=1,\ldots,k$.

In general, deciding whether a schedule exists for a Pinwheel Scheduling instance $A$ is NP-hard;
(see \wref{sec:related} for a thorough discussion of the complexity of the problem).

The \emph{density} of a Pinwheel Scheduling instance $A=\{a_1,\ldots,a_k\}$ 
is $d = d(A) = \sum_{i=1}^k 1/a_i$.
It is easy to see that $d(A)\le 1$ is a necessary condition for $A$ to be schedulable.
Any $A$ with $d(A) \le 1/2$ can be scheduled by rounding all frequencies down to the nearest power of 2 and assigning days greedily.
This ``density threshold'', \ie, a value $d^*$ so that every instance with $d\le d^*$ is schedulable, was successively improved in a sequence of papers from $d^*=0.5$ to
$d^* = 0.\overline 6$~\cite{largeClassPinwheel},
$d^* = 0.7$~\cite{pinwheelSchedulable2over3}, and finally
$d^* = 0.75$ in 2002~\cite{achievableDensities}.
Since the Pinwheel Scheduling instance $(2,3,M)$ is not schedulable for any $M$, 
$d^*=5/6 = 0.8\overline 3$ is the best we can hope for,
and Chan and Chin conjectured in 1993
that this is tight:

\begin{conjecture}[5/6 Conjecture~\cite{largeClassPinwheel}]
\label{conj:5-6-conjecture}
\ifmysiam{~\\}{}
	Every Pinwheel Scheduling instance with density $d\le \frac56$ is schedulable.
\end{conjecture}
No further progress on the gap between these general upper and lower bounds has been made for almost two decades.

We confirm \wref{conj:5-6-conjecture} for all Pinwheel Scheduling instances with $k\le 12$ tasks. This vastly expands recent work by Ding~\cite{5over6upto5tasks}, which achieved the same for up to 5 tasks using exhaustive manual case analysis.
The larger number of tasks substantially strengthens the confidence in the 5/6-conjecture since
these instances are a rich and diverse class of Pinwheel Scheduling problems,
and extending well beyond smaller, simpler cases.

Both~\cite{5over6upto5tasks} and this work are based on the observation that the infinitely many Pinwheel Scheduling instances with a fixed number $k$ of tasks actually fall into only a finite number of equivalence classes \wrt schedulability: the Pareto surfaces introduced in \wref{sec:pareto-theorem}.
Our works vastly differ in the methodology for finding these: 
Ding manually evaluates all possibilities, justifying independently for each possible
case that it either has density above $\frac56$ or admits a schedule.
We instead devise general algorithms to efficiently automate this task; 
our main achievement here is to substantially reduce the effort to verify the completeness of a Pareto surface: for $k=11$, from tens of thousands of calls to an oracle for an NP-hard problem to just 37 such calls.

Our result draws on a combination of structural insights about 
Pinwheel Scheduling 
and heavily engineered implementations of algorithms.
By systematically extending smaller instances to more tasks, an iterative algorithm
computes the Pareto surfaces for all instances with \emph{up to} $k$ tasks with overall dramatically fewer oracle calls than the method of~\cite{5over6upto5tasks} for a fixed value~$k$.

\needspace{3\baselineskip}

We further extend Ding's methodology to the set of \emph{all} Pinwheel Scheduling instances with $k$ tasks (instead those of density $d\le \frac56$).
We show that for any $k$, there is a \emph{finite} set $\mathcal C_k$ of schedules, so that any instance with $k$ tasks can be solved if and only if it can be solved by a schedule from $\mathcal C_k$,
and we give certifying algorithms for computing $\mathcal C_k$.
Their running time grows very rapidly with $k$, 
but we give $\mathcal C_k$ up to $k=5$ (see \wtpref{tab:pareto-surfaces}).
The highly efficient backtracking algorithms for general Pinwheel Scheduling instances
developed as part of this research are of independent interest, 
both for Pinwheel Scheduling itself, as well as for the related Bamboo Garden Trimming problem~\cite{BGTIntro}.

\paragraph{Outline}

The remainder of this first section gives a more comprehensive discussion of related work and
the rest of the paper is organized into a theory part and an engineering part.
We first introduce basic notions about Pinwheel Scheduling in 
\wref{sec:prelims}, followed by the main theoretical results in \wref{sec:pareto-theorem}.
\wref{sec:algorithms} describes our engineered implemented oracles for single Pinwheel Scheduling instances, which determine feasibility and find schedules;
\wref{sec:5-6-surface} then describes our implementation of the Pareto-surface computation.
In \wref{sec:evaluation}, we report on a running-time study, comparing our algorithms and analyzing their efficiency.
We conclude with a summary of results and future work in \wref{sec:conclusion}.
\ifproceedings{%
    Some proofs are deferred to the appendix of the extended online version.
}{%
    \wref{app:proofs} contains a few proofs omitted from the main text.
    \wref{app:worked-examples} gives worked examples for the algorithms from \wref{sec:algorithms}.
}%

\subsection{Related Work}
\label{sec:related}

The Pinwheel Scheduling problem was originally proposed 
by Holte et al.~\cite{PinwheelIntro} in 1989 
in the context of assigning receiver time slots to satellites 
with varying bandwidth requirements which share a common ground station.
They introduce the notion of density, show that $d>1$ implies infeasibility, and give the algorithm to schedule any instance with $d\le\frac12$.
A sequence of papers~\cite{largeClassPinwheel,pinwheelSchedulable2over3,achievableDensities}
extended this result to all instances of density at most $\frac34$.

A second line of research aims to confirm \wref{conj:5-6-conjecture} for restricted
classes of instances.
Efficient algorithms for computing schedules are sometimes also considered; 
here the complication that exponentially long periodic schedules can be necessary 
lead to the introduction of ``fast online schedulers'' as output, 
\ie, a simple program that can produce the schedule on demand~\cite{PinwheelIntro}.
Closest to our work is a recent article by Ding~\cite{5over6upto5tasks},
who confirmed \wref{conj:5-6-conjecture} for instances with $k\leq 5$ tasks
through manually determining a Pareto trie (in our terminology) of instances with $d\le \frac56$.

An orthogonal line of work considered all Pinwheel Scheduling instances with a
fixed number of \emph{distinct} values for frequencies (but an arbitrary number of tasks):
first for $2$ distinct frequencies~\cite{PinwheelIntro,pinwheel2numbers} and later for $3$~\cite{3Task5over6}.
In each case, the approaches taken seem unsuitable for extension beyond the scenarios studied therein.

Various generalizations of Pinwheel Scheduling have also been studied,
for example dropping the requirement of unit-length jobs~\cite{HanLin1992,generalisedPinwheelScheduling}.

The complexity status of Pinwheel Scheduling has gained some notoriety in the literature.
Holte et al.~\cite{PinwheelIntro} showed that the problem is in PSPACE; whether it is in NP
is not obvious since exponentially long periodic schedules can be necessary, so a standard approach of nondeterministically guessing a witness for feasibility does not have polynomial runtime.
Holte et al.\ further stated that the problem is NP-hard in compact encoding, \ie, 
when all tasks of the same frequency are encoded as a pair of integers
(the frequency and the number of such tasks)
but they postponed the proof to a follow-up article that seems not to have been published.

Let Exact-Pinwheel Scheduling be the variant of the Pinwheel Scheduling problem 
where a schedule is only valid if two consecutive executions of task $i$ are \emph{exactly}
$a_i$ days apart.
Bar-Noy et al.~\cite[Thm. 13]{bar2002minimizing} show that this problem is NP-complete (they refer to it as Periodic Maintenance Scheduling) by a reduction from Graph Colouring.
Later, Bar-Noy et al.~\cite{windowsSchedulingBinPacking} observe that we can also reduce
Exact-Pinwheel Scheduling to the special case of standard Pinwheel Scheduling of dense instances by filling up an instance $A$ with exact frequencies $a_1,\ldots,a_k$ with as many tasks of exact frequency $\operatorname{lcm}(a_1,\ldots,a_k)$ as needed to reach density $1$.
Finally, on dense instances, exact and standard (upper-bound) frequencies are equivalent.
Together, this proves that Pinwheel Scheduling in compact encoding is indeed NP-hard.

In an arxiv preprint from 2014, Jacobs and Longo~\cite{pinwheelNPhard} strengthened these results to prove NP-hardness for 
Pinwheel Scheduling in standard representation, \ie, where the frequencies are simply encoded as a sequence. They also claim a reduction to instances with maximal period length
in $n^{O(\log n \log \log n)}$, indicating that even pseudo-polynomial algorithms for Pinwheel Scheduling are unlikely to exist. These results do not seem to appear in a peer-reviewed venue.

To our knowledge, known reductions only lead to instances of density $d=1$;
whether Pinwheel Scheduling remains NP-hard for instances with density $d<1$ seems yet to be determined.

\section{Preliminaries}
\label{sec:prelims}

In this section, we define some core notation used throughout this paper,
and we collect some facts about Pinwheel Scheduling.
Most of these have appeared in previous work, but the proofs are so short that we prefer to give a self-contained presentation.

It will be convenient to slightly extend schedules to also allow a special symbol ``--'',
which means that \emph{no} task is executed on that day.
We refer to these days as \emph{holidays} or \emph{gaps} in the schedule.
Clearly, any holidays in a valid schedule could be filled with an arbitrary task without 
affecting its validity.

If a Pinwheel Scheduling instance $A=(a_1,\ldots,a_k)$ satisfies $a_k = a_{k-1} = \cdots = a_{k-\ell+1}\ne a_{k-\ell}$, we call $\ell$, the number of tasks of equal maximal frequency, the \emph{symmetry} of $A$.

We call $A$ \emph{dense} if its density is $d=1$.

Let $A=(a_1,\ldots,a_k)$ be a Pinwheel Scheduling instance with valid infinite schedule $S_\infty$.
For any day $t$ in a schedule $S_\infty$, the \emph{state}, $X=X(t)$, of a Pinwheel Scheduling instance is a vector $X=(x_1,\ldots,x_k)$, where $x_i$ is the number of days in the schedule since the last occurrence of $i$,
$x_i = t-\max\{t'\le t : s_i(t') = i\}\cup\{0\}$.
Note that since $S_\infty$ is valid, all states it reaches must also be \emph{valid} ($0\le x_i< a_i$, for $i\in[k]$).
This condition also implies that there are only finitely many valid states.

This notion of states allows us to cast Pinwheel Scheduling to a graph problem.
Define the \emph{state graph} $G_A=(V_A, E_A)$ as the directed graph with all possible states as vertices, \ie,
$V_A = \{(x_1,\ldots,x_k) : \forall i\in[k]\: 0\le x_i < a_i \}$, and an edge $(X,Y) \in E_A$ 
\begin{enumerate}
    \item if $\exists j\in[k] : \bigl( y_j = 0 \rel\land \forall i\in[k]\setminus\{j\} : y_i = x_i+1 \bigr)$
    	(\textit{task edges}),
    \item if $\forall i\in[k] : y_i = x_i+1$
    	(\textit{gap edges}), or 
    \item if $X=X_0\ne Y$ where $X_0=(0,\ldots,0)$ (\textit{start edges}).
\end{enumerate}
Then $A$ is schedulable if and only if $G_A$ contains an infinite walk starting at $X_0 = (0,\ldots,0)$.
Since $G_A$ is finite and we have the start edges, 
$A$ is schedulable if and only if $G_A$ contains a (directed) cycle.
We call a state \emph{sustainable} if it can be revisited infinitely often by some valid schedule (\ie, when it is part of a directed cycle in $G_A$).

It further follows from the state-graph representation that if $A$ is at all schedulable,
it is so by a \emph{periodic} schedule, \ie, there is a schedule $S=s_1s_2\ldots$ and an integer $p$,
so that for all $t$ we have $s_t = s_t+p$.
Unless explicitly mentioned in the following we assume schedules to be periodic 
and we represent them by the finite periodic part, $S=s_1\ldots s_p$.
Since $p$ corresponds to the length of a cycle in $G_A$, we can always find $S$ with 
$p=|S| \le |V_A| = \prod_{i=1}^k a_i$ (cf.~\cite{PinwheelIntro}) if $A$ is schedulable.

\begin{table}[htb]
\centering\smaller[1]
\begin{tabular}{*{2}{ >{\(} l <{\)} }}
\toprule
\textbf{Instance} & \textbf{Schedule}         \\
\midrule
(1) & (1) \\
\midrule
(2,2) & (1,2) \\
\midrule
(2, 4, 4)     & (1, 2, 1, 3) \\
(3, 3, 3)     & (1, 2, 3)   \\
\midrule
(2, 4, 8, 8)     & (1, 2, 1, 3, 1, 2, 1, 4) \\
(2, 6, 6, 6)     & (1, 2, 1, 3, 1, 4) \\
(3, 3, 6, 6)     & (1, 2, 3, 1, 2, 4) \\
(3, 4, 5, 8)     & (1, 2, 4, 1, 3, 2, 1, 3) \\
(3, 5, 5, 5)     & (1, 2, 3, 1, 4, 2, 1, 3, 4) \\
(4, 4, 4, 4)     & (1, 2, 3, 4) \\
\midrule
(2, 4, 8, 16, 16)       & (1, 2, 1, 3, 1, 2, 1, 4, 1, 2, 1, 3, 1, 2, 1, 5) \\
(2, 4, 12, 12, 12)      & (1, 2, 1, 3, 1, 2, 1, 4, 1, 2, 1, 5) \\
(2, 6, 6, 12, 12)       & (1, 2, 1, 3, 1, 4, 1, 2, 1, 3, 1, 5) \\
(2, 6, 8, 10, 16)       & (1, 2, 1, 3, 1, 5, 1, 2, 1, 4, 1, 3, 1, 2, 1, 4) \\
(2, 6, 10, 10, 10)      & (1, 2, 1, 3, 1, 4, 1, 2, 1, 5, 1, 3, 1, 2, 1, 4, 1, 5) \\
(2, 8, 8, 8, 8)         & (1, 2, 1, 3, 1, 4, 1, 5) \\
(3, 3, 6, 12, 12)       & (1, 2, 3, 1, 2, 4, 1, 2, 3, 1, 2, 5) \\
(3, 3, 9, 9, 9)         & (1, 2, 3, 1, 2, 4, 1, 2, 5) \\
(3, 4, 5, 14, 14)       & (1, 2, 3, 1, 4, 2, 1, 3, 1, 2, 5, 1, 3, 2) \\
(3, 4, 6, 10, 16)       & (1, 2, 3, 1, 4, 2, 1, 3, 1, 2, 4, 1, 3, 2, 1, 5) \\
(3, 4, 6, 11, 11)       & (1, 2, 3, 1, 5, 2, 1, 3, 2, 1, 4) \\
(3, 4, 8, 8, 8)         & (1, 2, 4, 1, 5, 2, 1, 3) \\
(3, 5, 5, 9, 9)         & (1, 2, 5, 1, 3, 2, 1, 4, 3) \\
(3, 5, 6, 7, 12)        & (1, 2, 4, 1, 3, 2, 1, 4, 2, 1, 3, 5) \\
(3, 5, 7, 7, 9)         & (1, 2, 3, 1, 4, 2, 1, 5, 3, 1, 2, 4, 1, 3, 2, 1, 5, 4) \\
(3, 5, 7, 8, 8)         & (1, 2, 3, 1, 4, 2, 1, 5, 1, 3, 2, 1, 4, 5) \\
(3, 6, 6, 6, 6)         & (1, 2, 3, 1, 4, 5) \\
(4, 4, 4, 8, 8)         & (1, 2, 3, 4, 1, 2, 3, 5) \\
(4, 4, 5, 7, 12)        & (1, 2, 3, 4, 1, 2, 5, 3, 1, 2, 4, 3) \\
(4, 4, 6, 6, 6)         & (1, 3, 2, 4, 1, 5, 2, 3, 1, 4, 2, 5) \\
(4, 5, 5, 6, 10)        & (1, 2, 3, 5, 1, 4, 2, 3, 1, 4) \\
(4, 5, 5, 7, 7)         & (1, 2, 5, 3, 1, 4, 2, 1, 3, 5, 2, 1, 4, 3) \\
(5, 5, 5, 5, 5)         & (1, 2, 3, 4, 5) \\
\bottomrule
\end{tabular}
\caption{The Pareto surfaces $\mathcal C_k$ for $k\le 5$.}
\label{tab:pareto-surfaces}
\end{table}

Given two Pinwheel Scheduling instances $A=(a_1,\ldots,a_k)$ and $B=(b_1,\ldots,b_k)$
(with the same number of tasks $k$),
we say that $A$ \emph{dominates} $B$, written $A\le B$, if $\forall i\in[k]: a_i\le b_i$.
Obviously, any schedule that is valid for $A$ is also valid for $B$.
Moreover, $A\le B$ implies $d(A) \ge d(B)$.

We call a schedulable Pinwheel Scheduling instance $A$ \emph{loosely schedulable}
if it admits a periodic schedule with a gap (\ie, when $G_A$ has a cycle containing a gap edge);
otherwise $A$ is \emph{tight / tightly feasible}.
For example, $(2,4)$ is loosely schedulable by $(1,2,1,-)$, whereas 
$(2,3)$ is tightly feasible despite having density $5/6 < 1$; 
observe that $(2,3,*)$ is not schedulable for any value of~$*$.

\begin{proposition}
\label{pro:oracle}
	Given a Pinwheel Scheduling instance $A=(a_1,\ldots,a_k)$,
	it can be decided whether $A$ is infeasible, tightly schedulable or
	loosely schedulable using $O(k \prod_{i=1}^k a_i)$ time and space.
	Moreover, if it exists, a corresponding schedule can be computed with the same complexity
	and has length at most $\prod_{i=1}^k a_i$.
\end{proposition}

\needspace{3\baselineskip}

\begin{proof}
	We construct the state graph $G_A$ and compute the strongly connected components of~$G_A$.
	Note that $G_A$ contains a directed cycle iff there is a strong component containing at least two vertices; moreover, $G_A$ contains a directed cycle containing a gap edge
	iff there is a gap edge with both endpoints in the same strong component.
	Both the computation of strong components and testing these two conditions can be done in time linear in the size of $G_A$.
	We have $|V_A| = \prod_{i=1}^k a_i$ vertices in $G_A$.
	Since each vertex apart from $X_0$ in $G_A$ has at most $k$ outgoing edges
	and $X_0$ has at most $|V_A|$ outgoing edges, $|E_A|\le (k+1)|V_A|$.
	A schedule can be found using another depth-first search.
\end{proof}

We point out that the algorithm sketched in the proof above is mostly of theoretical interest
due to its prohibitive space cost.
We present several alternatives in \wref{sec:algorithms}.

\subsection{Small Frequency Conjecture}
\label{sec:conjectures}

We propose below two new conjectures about the structure of schedulable 
Pinwheel Scheduling instances.
These arose from observations made while engineering our algorithms, 
but are of independent interest.
We list evidence in their support in \wref{sec:5/6Optimisations}.

\begin{conjecture}[$2^k$ Conjecture]
\label{conj:2^k-conjecture}
\ifmysiam{~\\}{}\ifsiam{~\\}{}
	Let $A = (a_1,\ldots,a_k)$ be a \thmemph{loosely} schedulable Pinwheel Scheduling instance.
	Then $A$ admits a schedule $S$ with a holiday at least every $2^k$ days.
\end{conjecture}

\begin{conjecture}[Kernel Conjecture]
\label{conj:kernel-conjecture}
\ifmysiam{~\\}{}\ifsiam{~\\}{}
	Let $A = (a_1,\ldots,a_k)$ be a schedulable Pinwheel Scheduling instance.
	Then there exists another Pinwheel Scheduling instance $A'=(a_1',\ldots,a_k')$ such that:
	
	\begin{enumerate}[label=(\alph*)]
	    \item $A'$ is also schedulable,
	    \item $A'$ dominates $A$, $A' \le A$, and
	    \item $a_k' \le 2^{k-1}$.
	\end{enumerate}
\end{conjecture}
We show that these two conjectures are indeed \emph{equivalent}.

\begin{proposition}[Equivalent conjectures]
	\wref{conj:2^k-conjecture} and \wref{conj:kernel-conjecture} are equivalent.
\end{proposition}
\begin{proof}
	First assume \wref{conj:2^k-conjecture} holds true.
	Let $A = (a_1,\ldots,a_k)$ be an arbitrary schedulable instance.
	If $a_k\le 2^{k-1}$, we can set $A'=A$; so assume that the last $\ell\ge 1$ frequencies are 
	$a_k,a_{k-1},\ldots,a_{k-\ell+1} > 2^{k-1}$.
	Define $B=(a_1,\ldots,a_{k-\ell})$.
	$B$ is loosely schedulable since
	we can use the same schedule as for $A$ and replace all occurrences of $i > k-\ell$ by ``--''.
	By \wref{conj:2^k-conjecture}, $B$ then admits a gapped schedule $S$ with a gap every $2^{k-\ell}$ days.
	We claim that we can now set $A' = (a_1,\ldots,a_{k-\ell},2^{k-1},\ldots,2^{k-1})$,
	\ie, truncate all frequencies exceeding $2^{k-1}$ at $2^{k-1}$ and remain schedulable.
	We use $S$ and assign the $\ell$ tasks $k-\ell+1,\ldots, k$ in a Round-Robin fashion to the gap
	in $S$. This achieves frequency at most $\ell\cdot 2^{k-\ell}$, for each of these tasks.
	We check indeed $\ell\cdots 2^{k-\ell} \le 2^{k-1}$, for all $k$ and $\ell\in\N$.
	
	Now conversely assume that \wref{conj:kernel-conjecture} holds true.
	Let $A = (a_1,\ldots,a_k)$ be loosely schedulable.
	Let $S$ be a gapped schedule for $A$ and set $g$ to the frequency of the gap in $S$.
	Define $B = (b_1,\ldots,b_{k+1})$ with $b_i = a_i$ for $i\in[k]$ and $b_{k+1}=\max\{g,a_k\}$.
	$B$ is schedulable since we can replace the gap in $S$ by $k+1$ and obtain a schedule for $B$.
	Hence by \wref{conj:kernel-conjecture}, there is a schedulable instance
	$B' = (b_1',\ldots,b_{k+1}')$ with $b_i' \le b_i$ and $b_i' \le 2^k$.
	Let $S'$ be a schedule for $B'$. By replacing $k+1$ in $S'$ by ``--'', we obtain a valid
	gapped schedule for $A$ and since $b_{k+1}' \le 2^k$, the gap frequency in this schedule is at most $2^k$.
\end{proof}

In light of this, the evidence in support of \wref{conj:kernel-conjecture} that we provide in in \wref{sec:5/6Optimisations} equally supports \wref{conj:2^k-conjecture}.

\section{The Pareto Surface}
\label{sec:pareto-theorem}

\begin{figure*}
    \centering
    \includegraphics[width=.95\linewidth]{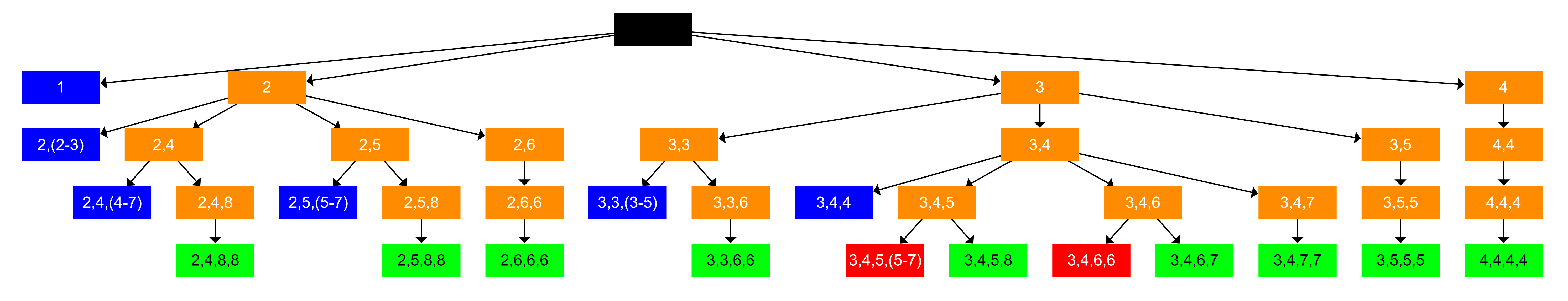}
    \caption{The portion of the Pareto trie $\mathcal T_4$ for $k=4$ that is explored by the traversal from \wref{sec:pareto-trie}. Nodes show the corresponding Pinwheel Scheduling instance or instances. The black root node is the empty Pinwheel Scheduling problem; blue nodes are tightly schedulable; orange nodes are loosely schedulable; red nodes are unschedulable and green nodes are schedulable, complete and form the Pareto surface. Solutions to the green nodes are shown in \wpref{tab:pareto-surfaces}.}
    \label{fig:pareto-surface-for-n=4}
\end{figure*}

We now derive our main new structural tool for analyzing Pinwheel Scheduling:
the notion of Pareto surfaces.
A special case of such a Pareto surface is (implicitly) used in~\cite{5over6upto5tasks}
(without developing its general applicability).

To this end, we need some more vocabulary.
It is often helpful to reduce Pinwheel schedules to their \emph{recurrence vectors}~-- the Pinwheel Scheduling instance solved by the schedule which minimizes $a_i$ for all $i$. To this end, we call a tuple consisting of a schedule and its recurrence vector a \emph{scheduled Pinwheel Scheduling problem}.
Let $\mathcal A$ be a (finite or infinite) set of Pinwheel Scheduling instances.
We say that a (finite or infinite) set of scheduled Pinwheel Scheduling problems $\mathcal S$ \emph{solves} $\mathcal A$
if, for every problem $A\in\mathcal A$, there is some $S\in\mathcal S$ so that $S$ includes a valid schedule for $A$.
This is equivalent to saying that every problem in $\mathcal A$ is dominated by some problem in $\mathcal S$.
A \emph{Pareto surface} $\mathcal C = \mathcal C(\mathcal A)$ 
for a set of Pinwheel Scheduling instances $\mathcal A$ is an inclusion minimal set of scheduled Pinwheel Scheduling problems that solves $A$, 
\ie, for every $A\in \mathcal A$ there is $C\in \mathcal C$ with $C\le A$.
We use \emph{inclusion minimal} to mean that no member of a set can be be removed from that set without violating its defining property, \ie, for every $C$ in $\mathcal C$ there most be some $A$ that is not solved by any other member of $\mathcal C$.
Note that while we consider only finite values of $\mathcal C(\mathcal A)$, $\mathcal A$ need not be finite, per \wref{thm:finite-pareto}.

The Pareto surfaces of two families of sets of Pinwheel Scheduling instances
are of particular interest:
$\mathcal P_k$, by which we denote the class of all Pinwheel Scheduling instances with $k$ tasks,
and $\mathcal P_{k,d}$, by which we denote the Pinwheel Scheduling instances with $k$ tasks and density at most $d$.
The main result of this section is the following theorem.

\begin{theorem}[Finite Pareto surfaces]
\label{thm:finite-pareto}
\ifsiam{~\\}{}
	For every $k\in\N$, there is a finite set of periodic schedules such that every Pinwheel Scheduling instance with $k$ tasks has a solution if and only if it has a solution in that set.
	Moreover, there is a unique inclusion-minimal such set~$\mathcal C_k$.
\end{theorem}

Before proceeding with the proof of this theorem in the following two subsections, let us note a complexity-theoretic consequence of this result.

\begin{corollary}[Pinwheel is FPT]
\label{cor:fpt}
\ifsiam{~\\}{}
	Pinwheel Scheduling is fixed-parameter tractable with respect to the number of tasks $k$.
\end{corollary}

Note that the \emph{input size} $N$ of a Pinwheel Scheduling instance can be substantially larger than $k$ since it has to encode the frequencies (say, in binary); 
frequencies at least exponential in $k$ are necessary even for just the instances 
in $\mathcal C_k$, and in general $N$ is not bounded in terms of $k$.

\begin{proofof}{\wref{cor:fpt}}
We give an algorithm deciding any instance $A\in\mathcal P_k$ in time $O(N+f(k))$ for $N$ the encoding length of $A$ and $f$ some computable function; this implies the claim.
By \wref{thm:finite-pareto}, all of $\mathcal P_k$ is solved by $\mathcal C_k$. 
Let $m(k)$ be the maximum over all distances between consecutive occurrences of all task in any of these solutions. 
We first compute $\mathcal C_k$ and $m(k)$; as $\mathcal C_k$ only depends on $k$, 
the cost to do so is bounded by some function $g(k)$. 
Read the input $A=(a_1,\ldots,a_k)$ (at cost $O(N)$), and replace any frequency $a_i > m(k)$ with $m(k)$, producing a new Pinwheel Scheduling instance $A'=\kappa(A)$ of (encoding) size $N' = O( k \log m(k) )$.
Now $A'$ is schedulable iff $A$ is schedulable, because any $S\in\mathcal C_k$ solves $A'$ iff it solves $A$.
Comparing $A'$ with all $S\in\mathcal C_k$ for some cost $h(k)$ determines whether there exists a schedule to $A$, so the schedulability of any input $A$ can be determined for a cost 
$O(g(k) + N + h(k))$. 
\end{proofof}

\begin{remark}[Kernel size]
	The construction above indeed shows that Pinwheel Scheduling has an FPT-kernel
	of size $O(k \log(m(k)))$; assuming \wref{conj:kernel-conjecture}, this reduces to $O(k^2)$.
\end{remark}

\subsection{The Pareto Trie}

\label{sec:pareto-trie}
Towards the proof of the first claim of \wref{thm:finite-pareto},
we describe an algorithm to compute the Pareto surface $\mathcal C_k$ for a given $k$,
based on an oracle for deciding whether a given Pinwheel Scheduling instance 
is infeasible, tightly schedulable or loosely schedulable (cf.~\wref{pro:oracle}).
We describe our implementation of such an oracle in \wref{sec:algorithms}.

This algorithm conceptually explores an (infinite) trie $\mathcal T_k$ for $\mathcal P_k$, where each node is the Pinwheel Scheduling instance that is a prefix of all of its descendants:
the root of $\mathcal T_k$ corresponds to the empty instance with no tasks at all.
It has infinitely many children, reached through edges labeled $1,2,3,\ldots{}$.
In general, every node $v$ at depth less than $k$ has infinitely many children; 
if $v$ is reached from its parent by an edge labeled $a$, $v$ has children $a,a+1,a+2,\ldots{}$.
We identify a node $v$ in the trie with the sequence of edge labels on the path from the root to $v$.
In this way, each node $v$ at depth $\ell$ corresponds to a Pinwheel Scheduling instance
on $\ell$ tasks.

Our algorithm explores $\mathcal T_k$ using a depth-first search.
Since $\mathcal T_k$ has depth $k$, we can only descend in the tree $k$ times;
however, we have to show that we only need to explore finitely many children of any node.
Suppose we are currently visiting a node $v$ at depth $\ell<k$, corresponding to a
Pinwheel Scheduling instance $A=(a_1,\ldots,a_\ell)$.
If $A$ is infeasible or tightly schedulable, all of $v$'s descendants are infeasible;
in particular, none of the descendants at depth $k$~-- corresponding to extensions of $A$ to instances in $\mathcal P_k$~-- is schedulable. So we need not visit any of them.


If $A$ is loosely schedulable, some descendant at depth $k$ is guaranteed to be feasible:
Let $S$ be a gapped schedule for $A$ with some gap frequency $g$.
Then $B=(a_1,\ldots,a_\ell, (k-\ell) g,\ldots, (k-\ell) g)$~-- 
\ie, $A$ with $k-\ell$ copies of $(k-\ell) g$ appended~--
is solved by aschedule $S'$ obtained from $k-\ell$ copies of $S$, with the gap replaced by $\ell+1,\ldots,k$, respectively.
Hence, we need not visit any child of $v$ with label larger than $(k-\ell) g$ (they are all feasible and dominated by $B$), and in particular, we only visit finitely many children of $v$.
For later reference, we call the smallest frequency $f$ so that $(a_1,\ldots,a_\ell,f,\ldots,f) \in \mathcal P_k$ 
is schedulable the \emph{Round Robin frequency} of $A$ with respect to $k$.
The root can be treated as a loosely schedulable node with a gap frequency $g=1$, depth $\ell=0$, and Round Robin frequency $k$. 
An example for the Pareto trie for $k=4$ is shown in \wref{fig:pareto-surface-for-n=4}.

As the Pareto trie is searched, an inclusion minimal subset of its leaves is maintained~-- after the search is complete, this will be $\mathcal C_k$.

\begin{remark}[Periodic solution length]
	We can prove a bound on $m(k)$, the largest frequency in any instance of $\mathcal C_k$, using the trie and \wref{pro:oracle}.
	For a loosely feasible instance $A=(a_1,\ldots,a_\ell)$, we can always find a schedule of length $\le \prod_{i=1}^{\ell} a_i$ with at least one gap, so 
	$g\le \prod_{i=1}^\ell a_i$. Hence, $A$'s Round Robin frequency $f$ is 
	$f\le (k-\ell)\prod_{i=1}^\ell a_i$, which gives an upper bound for all expansions $a_{\ell+1}$ to $A$ that can occur in~$\mathcal C_k$.
	We hence always have $m(k) \le m_k$, where
	$m_1=1$, $m_{\ell+1} = (k-\ell)\prod_{i=1}^\ell m_i$. 
	Since $m_k = \Omega(2^{2^k})$, this proven upper bound for $m(k)$ 
	is \emph{doubly} exponential in $k$, 
	whereas \wref{conj:2^k-conjecture} suggests 
	the singly exponential bound $m(k) \le 2^{k-1}$ is sufficient. 
	$m(k)\ge 2^{k-1}$ is also clearly necessary, so assuming \wref{conj:2^k-conjecture}
	would completely settle the question of the worst-case periodic solution~length.

\end{remark}

\subsection{Uniqueness}

Next we prove that $\mathcal C_k$ is unique.

\begin{lemma}[Characterization Pareto surface]
\label{lem:characterization-pareto}
	$A\in\mathcal C(\mathcal  P_k)$ iff $A$ is schedulable and decreasing any one component of $A$ by 1 makes the instance infeasible.
\end{lemma}
\begin{proof}
Let $A\in\mathcal C(\mathcal P_k)$~-- it is then schedulable by definition. 
Define $A_{i-} = (a_1,\ldots,a_{i-1},a_i-1,a_{i+1},\ldots,a_k)$.
Assume towards a contradiction that there is a task $i\in[k]$ so that $A_{i-}$ is also schedulable.
Since $A_{i-}$ is schedulable there must be some $B\in\mathcal C(\mathcal P_k)$ with $B\le A_{i-}$ and $B$ schedulable.
But then also $B\le A_{i-}\le A$, so anything dominated by $A$ is also dominated by $B$ and we can remove $A$ from the Pareto surface; a contradiction.

\needspace{3\baselineskip}

Let conversely $A$ be schedulable and for all $i\in[k]$, $A_{i-}$ is infeasible.
By definition of $\mathcal C(\mathcal P_k)$, there is a schedulable instance 
$B=(b_1,\ldots,b_k)\in\mathcal C(\mathcal P_k)$
with $B\le A$. Assume towards a contradiction that there is an index $i\in[k]$ with $b_i < a_i$;
then we have $B\le A_{i-}$. Since $A_{i-}$ is already infeasible, so must $B$ be; a contradiction.
\end{proof}

\needspace{3\baselineskip}

The uniqueness claim from \wref{thm:finite-pareto} now follows immediately from \wref{lem:characterization-pareto}:
$\mathcal C(\mathcal P_k)$ consists of exactly those instances that satisfy the condition of that lemma.

Note that $\mathcal P_{k,d}$ does not in general have a unique Pareto surface;
for example, $\mathcal P_{2,2/3}$ has Pareto surfaces $\{(2,2)\}$ and $\{(2,6),(3,3)\}$.
Both dominate all instances with density at most $2/3$ but fail to do so after deleting any one element of the sets, meeting the definition for $\mathcal C(\mathcal P_{k,d})$. In this instance, the former ($\mathcal C_2$) dominates the latter but that does not disqualify it as a Pareto surface. 
Clearly $\mathcal C_k$ is also a Pareto surface for $\mathcal P_{k,d}$, so finite
$\mathcal C(\mathcal P_{k,d})$ always exist.

\section{Engineering Pinwheel Scheduling}
\label{sec:algorithms}
\label{section: backtrackingAlgorithm}

In this section, we introduce our backtracking algorithm for general Pinwheel Scheduling instances, with three consecutive stages building on each other: the Naïve algorithm, the Optimised algorithm, and finally the Foresight algorithm. The effects of each optimisation are discussed in \wref{sec:evaluation}.
\ifproceedings{%
    Some correctness proofs are deferred to the appendix of extended online version, and the code is available online \cite{codeAlenex22}.
}{%
    Correctness proofs are given in \wref{app:proofs}.
}%

\subsection{The Naïve Algorithm}
\label{sec:Naive}
Each of the three algorithms presented here use a backtracking procedure to assess the schedulability of Pinwheel Scheduling instances. They form all possible solutions into a trie of \emph{candidate solution prefixes} ($S_c$), which they explore using four basic operations:

\begin{enumerate}
    \item Push, which appends the next unexplored letter to~\(S_c\).
    \item Pop, which deletes the last letter from \(S_c\).
    \item Failure testing, which tests whether the current state is valid.
    \item Success testing, which tests whether the current state is known to be sustainable.
\end{enumerate}

Whenever a node is reached, they test each for failure, then for success. If a node is invalid, the pop operation is employed until there is an unexplored letter to push. Nodes pass the success test if their state is the same as some ancestral node~-- the path from that node to this node is a solution $S$. If a node is valid, but not known to be sustainable the push operation is again employed. 
\ifproceedings{%
    A diagram of this procedure is shown in the appendix of the extended online version, along with several worked examples.
}{%
    A diagram of this procedure is shown in \wref{app:naive-worked-examples}, along with several worked examples.
}%

In the naïve algorithm, tasks are pushed in descending frequency order. In schedulable cases, this reduces the observed length of failed candidate solutions attempted before finding a viable schedule, thus reducing success testing cost. This seemed to reduce overall cost in many cases, probably because success testing is \(O(n^2)\) (where $n$ is the length of the testable solution fragment~-- naively $\len{(S_c)}$, but optimised in \wref{sec:soln-length}) and failure testing is \(O(k)\). 
\ifproceedings{%
    An example of this difference is described in the appendix of the extended online version.
}{%
    An example of this difference is described in \wref{fig:Naive 6,3,3}.
}%
Note that the first move can be freely chosen, as ultimately all tasks must be a part of the final schedule.

\subsection{The Optimised Algorithm}
\label{sec:optimisations}

The Optimised algorithm expands the Naïve algorithm described above with three improvements that remove repetitive and symmetric sections of the search space and reduce the cost of success testing.

\subsubsection{Repetition}
\label{subsubsection: repetition}

We first establish two simple properties when comparing different states.
If we consider two states of the system, $X$ and $X'$, then $X$ is \emph{worse} than $X'$ if no $x_i$ is less than the corresponding $x'_i$ and some $x_i$ is greater than the corresponding $x'_i$. Formally: \(\forall i: x_i \geq x'_i\) and \(\exists i\) such that \(x_i > x'_i\). $X'$ is then considered to be \emph{better} than $X$. 

\begin{lemma}
\label{lem:worse-schedulable}
    If some state \(X\) in a Pinwheel Scheduling instance $A$ is worse than another state \(X'\) in the same instance, then any valid schedule for $A$ starting in state \(X\) is also a valid schedule for $A$ starting in state \(X'\). 
\end{lemma}

\begin{lemma}
\label{lem:worse-unschedulable}
    If some state $X$ is worse than another state $X'$ and there exists no valid schedule that starts at $X'$ then there exists no valid schedule that starts at $X$.
\end{lemma}

Avoiding immediate repetitions which are not themselves solutions shrinks the search space without changing the schedulability of Pinwheel Scheduling problems.
This optimisation is based on the following observation.

\begin{proposition}[Repetition]
\label{pro:repetition}
    If an instance of Pinwheel Scheduling $A$ has a solution that contains an immediately repeated strict subsequence (\ie, for some sequence of letters $r$, a repetitive solution $S_r$ exists, such that $S_r = \ldots, r, r, \ldots \in \mathcal S$ and $r$ is not a solution to $A$) then there also exists a solution that does not contain that immediately repeated strict subsequence. 
\end{proposition}

The simplest exploitation of \wref{pro:repetition} considers the simplest possible repeated subsequence~-- single character repetitions. Forbidding these has a small benefit in unschedulable instances, namely reducing the effective alphabet size by 1 as the last letter played cannot be repeated.

In schedulable instances the observed effect was larger, because of the order in which the trie is explored. Due to fixed order conventions, immediately repetitive additions to candidate solution prefixes are often the first to be tried. If we divide the search space around the first schedule found ($S_1$), repetitive candidate solution prefixes are over-represented before $S_1$ and thus removing them has a stronger effect on schedulable instances.

\subsubsection{Frequency Duplication}
\label{sec:duplication}
We call tasks in a Pinwheel Scheduling instance $A$ with the same values of $a_i$ \emph{duplicates}, as they are indistinguishable until either task is performed 
(after this point they can be distinguished by their $x_i$ values, which can never again be identical).
Naïvely, duplicate tasks are distinguished by their order of appearance in $A$, but an alternative method exists which exploits frequency duplication by pruning identical subtrees.

\begin{proposition}[Duplicates]
\label{pro:duplication}
    If a Pinwheel Scheduling instance contains two tasks $i \ne i'$ with $a_i=a_i'$, then it is schedulable when $i$ is performed before $i'$ iff it is schedulable when $i'$ is performed before~$i$.
\end{proposition}

As having the same frequency is a transitive relation, \wref{pro:duplication} obviously applies to instances with more than two duplicate tasks. 
The algorithm can be optimised by choosing one ordering of all duplicate tasks, instead of naively exploring all orderings. 
While this effect is limited in scope (many Pinwheel Scheduling instances have no duplicates), it has a large effect on instances which have many duplicates.

\subsubsection{Minimum Solution Length}
\label{sec:soln-length}

Each candidate solution prefix $S_c$ has a \emph{composition formula} $R$, with $k$ components~-- each component $r_i$ representing the number of instances of $i$ in $S_c$. Let $L$ be the length of the candidate solution prefix, $L$ is given by $L = \sum_{i=1}^n r_i$. For a candidate solution prefix to be sustainable, each letter $i$ must appear at least every $a_i$ letters, so over the whole solution $r_i \geq \ceil{\frac{L}{a_i}}$. To calculate the minimum value of $L$, $L_{min}$, we start by setting $r_i = 1$ for all $i$, then increment $r_i$ for each $i$ value until this condition is simultaneously met for all $i$.
 
$L_{min}$ can be used to avoid unnecessary comparisons in success testing in 2 ways:
\begin{enumerate}
    \item Only comparing states which are $L_{\mathit{min}}$ apart, because no closer states can be identical.
    \item Only performing success testing when $\len{(S)} \geq L_{\mathit{min}}$.
\end{enumerate}

The former reduces the cost of success testing each node, while the latter reduces the number of nodes which perform success testing. This effect is significant because the cost of testing for success grows as \((\len{(S)})^2\) while all other costs remain constant over \(\len{(S)}\). For a discussion of minimum solution lengths in Pinwheel Scheduling instances with two distinct numbers, including several minimum solution length algorithms, see \cite{pinwheel2numbers}.

\subsection{The Foresight Algorithm}
\label{sec:Foresight}

This optimisation modifies the Naïve failure testing process to gain more information from a similar amount of work. Instead of tracking state $X(t)$, consider \emph{urgency} $U(t)$:
\begin{equation}
    \forall i, t: u_i(t) = a_i - x_i(t) - 1
    \label{equation: urgency}
\end{equation}

This requires a different procedure when a task is performed (to perform task $i$, set $u_i = a_i-1$) and a different growing procedure (to grow $V$, set $u_i(t+1) = u_i(t)-1$ for all $i$). 
The Naïve failure testing procedure would test that $\forall i:u_i\geq 0$ but an alternative failure testing procedure is now possible if the urgency values of tasks are stored in ascending order ($\forall i: u_i \leq u_{i+1}$):

\begin{proposition}[Urgency]
\label{pro:urgency}
    If an urgency state $U$ is schedulable, $\forall i: u_i\geq i$.
\end{proposition}

This can be used to detect failure up to to $k$ days in advance for little additional cost over the Naïve failure testing method, reducing tree height. It can also be used to \emph{force} the execution of certain tasks on certain days:

\begin{proposition}[Forcing]
\label{pro:forcing}
    If a schedulable urgency state $U$ exists such that $\exists i': u_{i'}=i'$, then the task at position $i'$ and all preceding tasks must be executed in the next $i'$ days.
\end{proposition}

This can be used to greatly restrict branching and hence tree breadth; if $\exists i'$ such that $u_{i'} = i'$ then the next move must have $i \leq i'$. This optimisation is compatible with all three changes described in \wref{sec:optimisations}, and all three are included in the final Foresight implementation.

\subsection{Deciding Tight Feasibility}
As introduced in  \wref{sec:prelims}, the tightness of Pinwheel Scheduling instances can be determined by testing for the existence of a schedule containing at least one gap. This was implemented using the Optimised algorithm in \wref{sec:optimisations}, by making the default action from every position a holiday and adding an extra testing step after a sustainable state was found~-- searching the schedule that produced this state for a holiday. If no gap is found in that schedule, the search continued until a loose schedule was found or it was demonstrated that no loose schedule can exist. In principle, it would be possible to implement the Foresight algorithm from \wref{sec:Foresight} with gaps, but this would have required a full re-implementation of Foresight~-- a substantial time investment.

\section{Engineering the 5/6 Surfaces}
\label{sec:5-6-surface}

Our principal application of the algorithms from \wref{sec:algorithms} is the investigation of the $\frac{5}{6}$ conjecture for low $k$ values.
This section describes our algorithm for computing a Pareto surface for $\mathcal C(\mathcal P_{k,5/6})$, code for which is available online \cite{codeAlenex22}.

\subsection{Core algorithm}
We search the trie of Pinwheel Scheduling problems introduced in \wref{sec:pareto-trie} using the depth first search procedure outlined in that section. The search from a node at depth $h$ begins by creating a child with the smallest possible added frequency, then proceeds until the subtrie of the new node is fully explored. If a node has density $d=\frac{5}{6}$, it can have no descendants and is fully explored~-- otherwise the depth first search proceeds by fully exploring all children until each has a descendent with a symmetry of $k + 1 - h$. 
This descendent necessarily dominates all siblings seen after it in a depth first search. As only nodes with a density $d\leq\frac{5}{6}$ need be considered, denser nodes are ignored by this process.

We will outline several optimisations which introduce denser problems that may be used to dominate problems found by this search. 
To show that the $\frac{5}{6}$ conjecture is true for a certain value of $k$, we need to show that no unschedulable Pinwheel Scheduling systems with a density $\leq\frac{5}{6}$ exist for that value of $k$. That is, we need to show that the set of all unschedulable Pinwheel Scheduling systems with $d\leq \frac{5}{6}$ found when constructing $\mathcal C(\mathcal P_{k,5/6})$ is the empty set.

\subsection{Constructing the Pareto Surface}
\label{sec:5/6Optimisations}

We could consider \emph{the} density restricted Pareto surface comprised exclusively of members of $\mathcal P_{k,5/6}$, but this would prevent many useful optimisations. Instead, we require \emph{any} density restricted Pareto surface, consisting of a set of solutions which solve all Pinwheel Scheduling systems with a density $\leq\frac{5}{6}$~-- that is, we allow our surface to contain denser instances, so long as it remains complete and no unschedulable Pinwheel Scheduling systems with density $\leq\frac{5}{6}$ are found. This allows for optimisations which use easily schedulable Pinwheel Scheduling systems to dominate large classes of non-trivially schedulable Pinwheel Scheduling systems with density below~$\frac{5}{6}$.


\subsubsection{Frequency Capping}
This optimisation builds on \wref{conj:kernel-conjecture}, which we eagerly assume to be true here but then immediately check the validity of for each instance generated.
We cap the maximum $a_i$ value of considered Pinwheel Scheduling instances at $2^{k-1}$. This only lowers frequencies, so the capped Pinwheel Scheduling instance dominates both the instance it was created from and often many similar instances~-- particularly when multiple frequencies are capped. 

Because capping reduces frequencies, it raises densities. To avoid a potential issue where the density of a problem is below $\frac{5}{6}$ before capping but above $\frac{5}{6}$ after capping, we replace $\mathcal P_{k,5/6}$ with the similar and dominant $\mathcal P^*_{k,5/6}$. This set includes all problems where either $d\leq \frac{5}{6}$ and $\forall i: a_i < 2^{k-1}$ or which consist of a prefix with density $d\leq \frac{5}{6}$ and a suffix where $\forall i: a_i = 2^{k-1}$.

Any unschedulable members of $\mathcal P^*_{k,5/6}$ would be counterexamples to either \wref{conj:5-6-conjecture} or \wref{conj:kernel-conjecture}; which one would require future investigation.
Both conjectures have proved true in all presently considered instances.

\subsubsection{Folding}
This optimisation uses a pair of simple operations on Pinwheel Scheduling instances: folding and unfolding.
The \emph{$c$-task folding} of a Pinwheel Scheduling instance $A = (a_1,\ldots,a_k)$
is the instance 
$B = (b_1,\ldots,b_{k-c+1})$ with $k-c$ tasks with respective frequencies $a_1,\ldots,a_{k-c}$ and one task of frequency $\floor{a_{k-c+1}/c}$,
\ie, $B$ is obtained from $A$ by replacing the last $c$ tasks by a single one with frequency
$\floor{a_{k-c+1}/c}$.
The \emph{$c$-wise unfolding} of a Pinwheel Scheduling instance $A = (a_1,\ldots,a_k)$ \emph{at task $i$} is the instance $B=(b_1,\ldots,b_{k+c-1})$, where $B$ has $k-1$ tasks of frequencies $a_1,\ldots,a_{i-1},a_{i+1},\ldots,a_k$ plus $c$ tasks each with frequency $c a_i$.

Note that unfolding does not alter density, whereas folding never \emph{lowers} density
(but can substantially increase it).
Moreover, any schedule for $A$ can be turned into a schedule for a $c$-wise unfolding of $A$ by repeating the schedule $c$ times, replacing $i$ each time by a different copy in $B$.
Likewise, any schedule for a $c$-task folding of $A$ can be used to generate a schedule for $A$ itself by the same process.

We use this as follows.
Whenever a Pinwheel Scheduling instance needs to be solved, we try to find schedules for all viable foldings of that instance in parallel. 
If any $c$-task folding is schedulable, we find the strictest Pinwheel Scheduling instance solved by its schedule and unfold the folded task back into $c$ tasks. 
If $c>1$, this unfolded instance will dominate the original instance $A$~-- it will usually have higher symmetry than $A$. It will often be faster to solve than $A$, because it has both fewer tasks and smaller task separations. 

No challenge to the $\frac{5}{6}$ conjecture has been found unless the original instance $A$ is unschedulable, so all foldings can be considered in parallel and any unschedulable instances with $c>1$ discarded. 
As some Pinwheel Scheduling instances can be dramatically more challenging to solve than others (for our tools), a very substantial speedup was achieved by running all foldings in parallel and terminating all threads as soon as the first schedulable folding was found.

\begin{table*}[bht]
        \centering
        \begin{tabular}{r|ccc | cc | c}
        \toprule
            $k$  & Foresight            & Optimised            & Naïve                  & Opt/FS              & Naïve/Opt         & Surface Size \\
            \midrule
            6  & 2.36 $\pm$ 0.02        & 3.30 $\pm$ 0.07      & 4.97 $\pm$ 0.06        & 1.40 $\pm$ 0.03     & 1.51 $\pm$ 0.04   &   23    \\
            7  & 6.65 $\pm$ 0.04        & 18.80 $\pm$ 0.07     & 39.5 $\pm$ 0.3         & 2.83 $\pm$ 0.02     & 2.10 $\pm$ 0.02   &   78    \\
            8  & 16.3 $\pm$ 0.2         & 487 $\pm$ 3          & 895 $\pm$ 6            & 29.8 $\pm$ 0.3      & 1.84 $\pm$ 0.02   &   214   \\
            9  & 105.0 $\pm$ 0.8        & 367 $\pm$ 3          & 645 $\pm$ 3            & 3.50 $\pm$ 0.04     & 1.76 $\pm$ 0.02   &   638   \\
            10 & 869 $\pm$ 5            & 944 $\pm$ 2          & 3130 $\pm$ 30          & 1.086 $\pm$ 0.006   & 3.31 $\pm$ 0.04   &   5347  \\
            11 & 4300 $\pm$ 20          & 4670 $\pm$ 10        & 9860 $\pm$ 60          & 1.015 $\pm$ 0.006   & 2.26 $\pm$ 0.02   &   15265 \\

        \bottomrule
        \end{tabular}
    \caption{Total time in seconds and relative speedup to generate the $5/6$ Pareto surface using three Pinwheel Scheduling oracles and all optimisations from \wref{sec:5/6Optimisations}; the last column shows the size of the Pareto surface found by Foresight, taken from a representative run because errors are too small to adequately estimate. Data for $k=12$ is not shown, as different hardware was used to generate that surface and this data.}
    \label{tab:5/6PWS performance}
\end{table*}

\subsubsection{Initializing the Pareto Surface}
\ifsiam{~\\}{}
While the core algorithm constructs $\mathcal C(\mathcal P_{k,5/6})$ from scratch,
we can speed this up substantially by starting with some scheduled Pinwheel Scheduling problems, which may then be used to dominate problems in need of a solution. 

Unfolding a schedulable smaller instance to $k$ tasks is an easy way to generate
many scheduled Pinwheel Scheduling problems.

We pump prime our computations by using an \emph{unfolding surface}:
Consider the example of the one-task instance $(1)$. We recursively unfold this problem in each way that obtains $k$ tasks. All resulting instances are both dense and schedulable.
If we want instances of $k=3$ tasks, we first unfold $(1)$ to $(3,3,3)$ directly, then to $(2,2)$ and unfold that to $(2,4,4)$. 

The first version of this optimisation ($P(1)$) uses this unfolding of the single task instance $(1)$ as sketched above; the second ($P(5)$) unfolds $\mathcal C_5$, the Pareto surface for $k = 5$ and the third ($P(k-1)$) adds all elements of the \emph{previous} Pareto trie. These developmental stages are evaluated in \wref{sec:evaluating folding}.

\subsection{Searching the Pareto Surface}
Due to the success of the above approximations, the vast majority of Pinwheel Scheduling problems considered at any $k$ value have known solutions (99.8\% of problems in the Pareto surface for $k=11$ were made by unfolding the $k=10$ surface). If they are already known, solutions to each problem must be chosen from the known portion of the Pareto surface (which is large at high $k$ values, per \wref{tab:5/6PWS performance}), so efficiently searching this surface is crucial. 
This is effectively a dominance query over a dynamic set of points in $\R^k$.

While we could employ a standard data structure for orthogonal range searching here, the specific structure of our point set (the Pinwheel instances) suggests a bespoke trie-based solution:
We maintain all members of the known portion of the Pareto surface in a set of tries, separated according to their symmetries (each a subset of the Pareto trie, and using the same ordering conventions). These tries are searched in descending order of symmetry with depth first dominance queries, to find the highest symmetry solution to each solved problem. This data structure minimises repeated comparisons of the same value, but also exploits the structure of solvable Pinwheel Scheduling problems. 

When searching for a problem with a low value of $a_0$, high $a_0$ problems are excluded immediately. When searching for a problem with a higher $a_0$ value, low $a_0$ problems are initially considered, but eliminated quickly because they must escalate rapidly due to density constraints. The latter parts of problems are far less predictable, while also having a much larger range of possible values, and thus the dimensions which must be searched are highly asymmetric.


\section{Performance Evaluation}
\label{sec:evaluation}

In this section, we report on an extensive running-time study for various aspects of our tools.

\subsection{Pinwheel Schedulers}
\label{sec:pinwheelSchedulers}
We begin by evaluating the relative performance of an implementation of the state-graph based algorithm (Graph), as well as our backtracking algorithms (Naïve, Optimised and Foresight) using synthetic data. We then further evaluate the Naïve, Optimised and Foresight algorithms on the computation of Pareto surfaces.
\subsubsection{Randomly Generated Data}
\label{sec:randomlyGeneratedData}
The four schedulers introduced in \wref[Sections]{sec:prelims} and~\ref{sec:algorithms} were evaluated using Pinwheel Scheduling instances generated using the following random process: 
Let a real number $b$ be a density budget, initially 1. Generate a random real number $0 < r\leq b$, and from it a candidate task $a_c = \floor{\frac{1}{r}}$. While $b - \frac{1}{a_c} > 0$, continue adding new tasks to an initially empty Pinwheel Scheduling problem $A$, updating $b$ each time ($b \rightarrow \frac{1}{a_c}$). Once a task is rejected, if $b \neq 0$ add a final task $a_k = \ceil{\frac{1}{b}}$. Finally, sort $A$ and replace any tasks where $a_i > 2^{k-1}$ with $a_i = 2^{k-1}$. 

\needspace{3\baselineskip}

Problems generated this way have high densities and a mixture of low and high $a_i$ values, which makes them challenging to schedule. Capping by $2^{k-1}$ is done in light of \wref{conj:2^k-conjecture} to makes instances more representative of the problems our algorithms were designed for.

A running-time study using these problems is shown in \wref{fig:PWSTournament}, which demonstrates that each algorithm improves on its predecessor.
Here, we repeatedly draw instances from above distribution, but if an instance had already been drawn earlier, it is rejected and a new instance is drawn. 
Since instances with few tasks are more likely to arise in the above random process, non-rejected instances tend to get increasingly challenging over time. 
The correlation between instance difficulty for different algorithms is noteworthy. This suggests the existence of an intrinsic ``difficulty'' for Pinwheel Scheduling instances, at least \wrt our studied algorithms.
We leave a further exploration of this observation for future work.

\subsubsection{The 5/6 Surface}
\label{subsubsec: comparing using the 5/6 surface}

A secondary evaluation used the time it took to find the Pareto surface for $d \leq \frac{5}{6}$ using each method (see \wref{tab:5/6PWS performance}). This evaluation showed more variability between the performance of the Optimised and Foresight algorithm than seen in the previous section~-- with the Optimised algorithm taking $29.8 \pm 0.3$ times as long at $k=8$ but $1.015 \pm 0.006$ times as long at $k=11$ as the Foresight algorithm. This is likely due to the dominance of \emph{search time} at high $k$ values~-- at $k=11$, Foresight and Optimised respectively spent $96.2\pm 0.6\%$ resp.\ $87\pm 4\%$ of their time matching problems to previously found schedules. 

As such, the size of the $P(k-1)$ approximation of the Pareto surface was the determining factor in these times~-- a complex effect of the properties of the specific solutions found by each method. Future algorithms will aim to produce solutions more capable of dominating many problems and less costly search procedures for the Pareto surface.

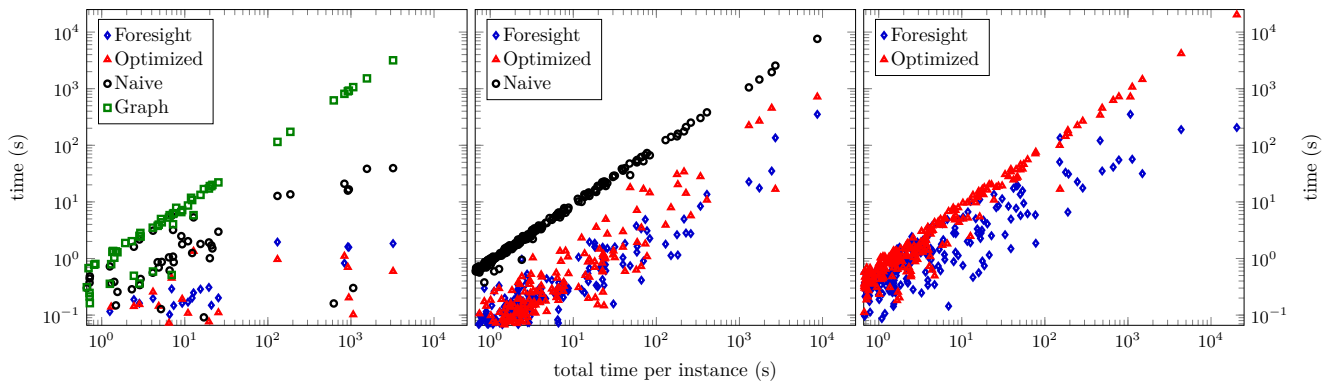
\begin{figure*}[thb]
	\externalizedpicture{plot-tournament}
	\hspace*{-.5em}{\resizebox{\linewidth+1em}!{\begin{tikzpicture}
	\pgfplotstableread[col sep=comma]{plots/fig-2.1.csv}\tablefigleft
	\pgfplotstableread[col sep=comma]{plots/fig-2.2.csv}\tablefigmiddle
	\pgfplotstableread[col sep=comma]{plots/fig-2.3.csv}\tablefigright
	\def\myxmin{.66}
	\def\myymin{0.066}
	\def\myxmax{25000}
	\def\myymax{25000}
	
	\begin{scope}[shift={(0,0)}]
	\begin{axis}[
			xmode=log,ymode=log,
			ymin=\myymin,ymax=\myymax,
			xmin=\myxmin,xmax=\myxmax,
			ylabel={time (s)},
			legend pos=north west, legend cell align=left,
			mark options={very thick,scale=.8},
		]
		\addplot[only marks,blue!80!black,mark=diamond] table[x=total,y=foresight] \tablefigleft ;
		\addlegendentry{Foresight}
		\addplot[only marks,red,mark=triangle] table[x=total,y=opt] \tablefigleft ;
		\addlegendentry{Optimized}
		\addplot[only marks,black,mark=o] table[x=total,y=naive] \tablefigleft ;
		\addlegendentry{Naive}
		\addplot[only marks,green!50!black,mark=square] table[x=total,y=graph] \tablefigleft ;
		\addlegendentry{Graph}

	\end{axis}
	\end{scope}
	
	\begin{scope}[shift={(7,0)}]
	\begin{axis}[
			xmode=log,ymode=log,
			ymin=\myymin,ymax=\myymax,
			xmin=\myxmin,xmax=\myxmax,
			xlabel={total time per instance (s)},
			legend pos=north west,legend cell align=left, 
			mark options={very thick,scale=.8},
			yticklabels=\empty,
		]
		\addplot[only marks,blue!80!black,mark=diamond] table[x=total,y=Foresight] \tablefigmiddle ;
		\addlegendentry{Foresight}
		\addplot[only marks,red,mark=triangle] table[x=total,y=Opt] \tablefigmiddle ;
		\addlegendentry{Optimized}
		\addplot[only marks,black,mark=o] table[x=total,y=Naive] \tablefigmiddle ;
		\addlegendentry{Naive}

	\end{axis}
	\end{scope}
	
	\begin{scope}[shift={(14,0)}]
	\begin{axis}[
			xmode=log,ymode=log,
			ymin=\myymin,ymax=\myymax,
			xmin=\myxmin,xmax=\myxmax,
			ylabel={time (s)},
			legend pos=north west, legend cell align=left,
			mark options={very thick,scale=.8},
			ylabel near ticks, yticklabel pos=right,
		]
		\addplot[only marks,blue!80!black,mark=diamond] table[x=total,y=Foresight] \tablefigright ;
		\addlegendentry{Foresight}
		\addplot[only marks,red,mark=triangle] table[x=total,y=Opt] \tablefigright ;
		\addlegendentry{Optimized}
	\end{axis}
	\end{scope}
	
	\end{tikzpicture}}}
	\caption{%
		Results of a tournament between the four Pinwheel Scheduling solvers introduced in \wref{sec:prelims} and \wref{sec:algorithms}, using the randomly generated Pinwheel Scheduling problems introduced in \wref{sec:randomlyGeneratedData}. After each round, the slowest method was eliminated: first the Graph method, then the Naïve method, and finally the Optimised method. The $x$-axes show the total time all methods took to solve each instance, a measure of overall complexity due to the apparent correlation between solve times for different methods. The $y$-axis shows the individual running time of the compared algorithms.%
	}
\label{fig:PWSTournament}
\end{figure*}

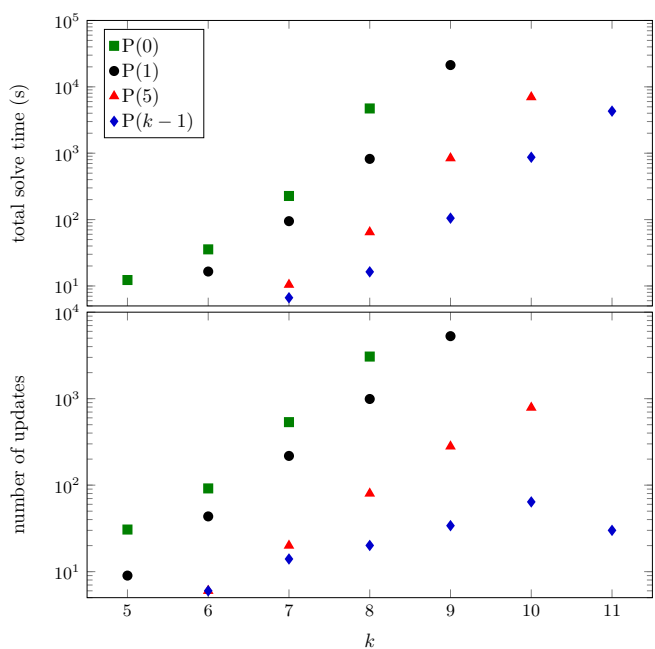
\begin{figure}[bh]

	\externalizedpicture{plot-approximations}
    \hspace*{-.5em}\resizebox{\linewidth+.5em}!{\begin{tikzpicture}
    \pgfplotstableread[col sep=comma]{plots/fig-3.1-time-taken.csv}\tablefigtop
    \pgfplotstableread[col sep=comma]{plots/fig-3.2-n-updates.csv}\tablefigbottom
   	\begin{scope}[shift={(0,0)}]
	\begin{axis}[
			xmode=normal,ymode=log,
			ymin=5,ymax=100000,
			xmin=4.5,xmax=11.5,
			ylabel={total solve time (s)},
			legend pos=north west, legend cell align=left,
			mark options={very thick,scale=1},
			width=1.4\linewidth,
			height=.8\linewidth,
			xticklabels=\empty,
		]
		\addplot[only marks,green!50!black,mark=square*] table[x=k,y=no approx] \tablefigtop ;
		\addlegendentry{P(0)}
		\addplot[only marks,black,mark=*] table[x=k,y=P(1)] \tablefigtop ;
		\addlegendentry{P(1)}
		\addplot[only marks,red,mark=triangle*] table[x=k,y=P(5)] \tablefigtop ;
		\addlegendentry{P(5)}
		\addplot[only marks,blue!80!black,mark=diamond*] table[x=k,y=P(k-1)] \tablefigtop ;
		\addlegendentry{P($k-1$)}
	\end{axis}
	\end{scope}
	
	\begin{scope}[shift={(0,-5.3)}]
	\begin{axis}[
			xmode=normal,ymode=log,
			ymin=5,ymax=10000,
			xmin=4.5,xmax=11.5,
			ylabel={number of updates},
			mark options={very thick,scale=1},
			width=1.4\linewidth,
			height=0.8\linewidth,
			xlabel={$k$},
		]
		\addplot[only marks,green!50!black,mark=square*] table[x=k,y=no approx] \tablefigbottom ;
		\addplot[only marks,black,mark=*] table[x=k,y=P(1)] \tablefigbottom ;
		\addplot[only marks,red,mark=triangle*] table[x=k,y=P(5)] \tablefigbottom ;
		\addplot[only marks,blue!80!black,mark=diamond*] table[x=k,y=P(k-1)] \tablefigbottom ;
	\end{axis}
	\end{scope}
    \end{tikzpicture}}

    \caption{%
    	Running time (top) and total updates to each approximation required to generate a complete density-restricted Pareto surface (bottom) with no approximation, and with all three approximations introduced in \wref{sec:5/6Optimisations}. The top figure demonstrates that while all costs are exponential with respect to $k$, P($k-1$) improves on P(5), which improves on P(1), which improves on the algorithm using no approximation. The bottom figure shows why~-- the approximation of the Pareto surface is significantly better for the later methods. 
	} 
    \label{fig:analysing approximations}
\end{figure}
\subsection{Constructing the 5/6 Pareto Surface}
\label{sec:evaluating 5/6 surface optimisations}
This section evaluates the methods used to generate $\mathcal C(\mathcal P_{k,5/6})$ which introduced in \wref{sec:5/6Optimisations}.

\subsubsection{Frequency capping}
\label{subsubsec:kernel optimisation}
The Kernel optimisation improved performance in two key ways: Firstly, it increased the symmetry of problems, thus reducing the number of problems that needed to be considered. 
Secondly, it reduced the largest $a_i$ values, which was particularly helpful for the deadline driven Foresight algorithm, which often ignores tasks with large $a_i$ values for long periods of time. Reducing the maximum $a_i$ value combated this, but our ongoing work will produce a version of Foresight more capable of handling arbitrarily large $a_i$ values. 

\subsubsection{Folding}
\label{sec:evaluating folding}
While folding had several benefits, the principal one was in exploiting the large variance between the cost of solving different Pinwheel Scheduling problems. Problems with smaller $k$ values, smaller maximum $a_i$ values and schedulable $A$ values are substantially faster then the converse~-- solving only the fastest of a set of problems that differ in these respects thus saves very considerable amounts of time. While starting all problems in the folded set increases overall work, these problems are solved in parallel so this does not translate to a substantial additional time cost.

In addition to often being faster to solve, problems which have been folded, solved and then unfolded usually have higher symmetries and lower $a_i$ values than the problems used to generate them and are therefore better at dominating other instances.

\subsubsection{Initializing the Pareto Surface}\mbox{}\\
Approximating the Pareto surface had significant effects on solve times. In addition to being very cheap to generate, schedules produced by approximation tend to solve problems with very high densities (as all unfoldings of a problem share the density of that problem) and high symmetries~-- they are thus ideal for dominating problems and reducing the size of the search space.

As shown in \wref{fig:analysing approximations}, each of the three approximations we considered substantially improves on its predecessor. With the introduction of the P($k-1$) approximation, the principal cost of finding the density restricted Pareto surface became the process of searching a list of schedules largely produced by approximation~-- we expect the effect of the P($k-1$) optimisation to increase when the search process is optimised in our future work.

\subsection{Searching the Pareto Surface}
\label{sec:evaluating searching}
Both trie-based searching and naive searching were implemented, with Trie-based searching running $8.80\pm0.05$ times faster at $k=11$, a substantial speed increase (though the performance gain was less substantial at smaller $k$ values, probably due to their smaller Pareto surfaces and the overheads inherent in a more complex data structure).

\section{Conclusion}
\label{sec:conclusion}

We presented new evidence for the $5/6$-density conjecture in Pinwheel Scheduling (\wref{conj:5-6-conjecture})
by engineering algorithms to compute a finite set of schedules that solves
any of the infinitely many solvable instances with at most $12$ tasks and $d\leq \frac{5}{6}$.
This substantially strengthens the confidence in the conjecture and has led
to new tools (theoretical and software) of independent interest 
for studying Pinwheel Scheduling.

Moreover,
we have constructed the full Pareto surfaces of Pinwheel Scheduling problems for $k \leq 5$, shown in \wref{tab:pareto-surfaces}, \ie, any Pinwheel Scheduling instance with at most $5$ tasks is schedulable if and only if one of the schedules listed in \wref{tab:pareto-surfaces} is valid for it. 

There are several avenues for future work.
Apart from settling the longstanding $5/6$-density conjecture,
confirming (or refuting) our new $2^k$ and kernel conjectures (\wref{conj:2^k-conjecture} and \wref{conj:kernel-conjecture})
about the largest ``effective'' frequencies would have interesting structural consequences for Pinwheel Scheduling.
Settling the complexity status of Pinwheel Scheduling for non-dense instances is another intriguing direction.

On the practical side, the Bamboo Garden Trimming problem introduced in~\cite{BGTIntro}
has recently received attention in an extensive experimental work~\cite{DEmidioDiStefanoNavarra2019} in the context of approximation algorithms.
Our Pareto surfaces for Pinwheel Scheduling immediately imply
similar equivalence classes for Bamboo Garden Trimming;
our corresponding results have been omitted due to space constraints.
The consequences of these results for approximate algorithms in Bamboo Garden Trimming deserve further exploration and are the subject of ongoing work.

\myacknowledgements

\bibliography{references}

\ifproceedings{}{
    \clearpage
    \appendix
    \part*{Appendix}
        
    \section{Omitted Proofs from \wref{sec:algorithms}}
\label{app:proofs}

In this appendix we collect the correctness proofs for our optimizations of the 
Pinwheel backtracking algorithms.

\begin{proofof}{\wref{lem:worse-schedulable}}
    Let \(X\) and \(X'\) be two arbitrary states in a Pinwheel Scheduling system \(A\) such that \(X\) is worse than \(X'\) and $S$ be a valid schedule for $A$, starting in state $X$.
    
     Proceed by induction over time $t$ as this schedule is executed simultaneously on two copies of $A$~-- one starting from $X$ and the other from $X'$:
    \begin{itemize}
        \item Base case: $\forall i:x'_{i}(t)\leq x_i(t)$.
        
        \item Inductive hypothesis: After $t$ days of the execution of $S$, $\forall i:x'_i(t)\leq x_i(t)$.
        
        \item Inductive step: On day $t$, all tasks grow: $\forall i:x_i(t+1) = x_i(t)+1$ and $\forall i:x'_i(t+1) = x'_i(t)+1$. The task $S_t$ is then executed: $x_{S_t}(t+1) = x'_{S_t}(t+1) = 0$. Neither step can make any $x_i(t+1)$ lower than any $x'_i(t+1)$.
    \end{itemize}
    
    Therefore $\forall i,t: x'_i(t)\leq x_i(t)$. As $S$ solves $A$ from $X$, it follows that $\forall i,t: x_i(t)\leq a_i$. Therefore, $\forall i,t: x'_i(t)\leq a_i$, i.e. any solution to $X$ is a solution to $X'$. 
\end{proofof}

\begin{proofof}{\wref{lem:worse-unschedulable}}
    Let $X$ and $X'$ be two arbitrary states in a Pinwheel scheduling system $A$ such that $X$ is worse than $X'$ and $X'$ is unschedulable. Consider an arbitrary schedule $S$ that starts at $X$ and assume towards a contradiction that $S$ is infinite~-- that is that $S$ is a valid schedule.

    As $S$ is infinite, executing $S$ from $X$ will never violate $\forall i: x_i\leq a_i$. As $X$ is better than $X'$, $\forall i: x_i\leq x'_i$, so $\forall i: x'_i\leq a_i$ and $S$ is also a solution to $X'$. However, $X'$ is unschedulable -- a contradiction. Therefore $S$ is must be finite. Therefore $X$ must be unschedulable.
\end{proofof}

\begin{proofof}{\wref{pro:repetition}}
    Let $S_r$ be a solution to a Pinwheel Scheduling problem $A$, and let $r$ be a repeated phrase within $S_r$, \ie, $S_r$ takes the form $S_r = \cdots r, r \cdots$. Let the state immediately following the first instance of $r$ be $X'$ and the state immediately following the second instance be $X$. This proof proceeds by induction over $\len(r)$:
    \begin{itemize}
        \item Base case: if $\len(r) = 1$ then $x'_{r_1}=x_{r_1}=0$ as the task $r_1$ was performed immediately before considering both states. If $k(A) = 1$ then $X = X'$, otherwise $\forall i \neq r_1: x_i = x'_i + 1$ as all executed tasks have grown. Therefore either $X = X'$ or $X$ is worse than $X'$.
        
        \item Inductive hypothesis: for $\len(r) = L$, either $X = X'$ or $X$ is worse than $X'$.
        
        \item Inductive step: If $X(t) = X'(t)$ then $\forall i:x_i(t) = x'_i(t)$ and extending $r$ will set $x'_{r_L}(t+1)=x_{r_L}(t+1)=0$ and $\forall i \neq r_L: x_i(t+1) = x'_i(t+1) = x_i(t)+1$, so $X(t+1) = X'(t+1)$. 
        
        If $X(t)$ is worse than $X'(t)$ then $\forall i: x_i(t) \geq x'_i(t)$. Extending $r$ will perform the task $r_L$ so that $x_{r_L}(t+1) = x'_{r_L}(t+1) = 0$ and let all other task separations grow ($\forall i \neq r_L: x_i(t+1) = x_i(t)+1$ and $\forall i \neq r_L: x'_i(t+1) = x'_i(t)+1$).
        
        Therefore either $X(t+1) = X'(t+1)$ or $X(t+1)$ is worse than $X'(t+1)$. 
    \end{itemize}
    This means that repeating a phrase $r$, of any length, either returns the system to the same state or takes it to a worse state. 
    
    If $r$ returns the system to the same state then $r$ is a solution to $V$. Else, per  \wref{lem:worse-schedulable} and \wref{lem:worse-unschedulable}, either $X'$ is unschedulable or there exists a solution from $X'$ which is also a solution from $X$.
\end{proofof}

\begin{proofof}{\wref{pro:duplication}}
    Consider a Pinwheel Scheduling system $A_s$ containing two symmetric tasks ($i$ and $i'$ such that $a_i = a_{i'}$). Let $S_c$ be an arbitrary sequence in $A_s$ where $a_i$ is performed before $a_{i'}$. Then let $S_c'$ be an arbitrary sequence identical to $S_c$, save that every occurrence of $i$ is replaced with an occurrence of $i'$ and vice versa. As $a_i = a_{i'}$, $S_c'$ exists iff $S_c$ exists. Therefore, if an infinite $S_c$ exists, an infinite $S_c'$ exists and vice versa. Also, if no infinite $S_c$ exists then no infinite $S'$ exists and vice versa. Therefore $A_s$ is schedulable when $i$ is performed before $i'$ iff $A_s$ is schedulable when $i'$ is performed before $i$.
\end{proofof}

\begin{proofof}{\wref{pro:urgency}}
    Consider a Pinwheel Scheduling instance $A$ at urgency state $U$ such that $A$ is schedulable from $U$. Proceed by induction over $i$.
    \begin{itemize}
        \item Base case: as $U$ is schedulable, $\forall i:u_i\geq 0$ so $u_0\geq 0$.
        \item Inductive hypothesis: $u_i \geq i$.
        \item Inductive step: $U$ is ordered, so $u_{i+1}\geq u_i$, hence $u_{i+1}\geq i$. If $u_{i+1} > u_i$ then $u_{i+1}\geq i+1$. Alternatively, if $u_i=u_{i+1}$ then either $u_i > i$ and $u_{i+1} = u_i > i$ or $u_{i+1} = u_i = i$. 
        
        Assume that $u_{i+1} = u_i = i$ towards a contradiction.
        
        As $U$ is schedulable, some infinite $S$ exists such that $\forall i,t:u_i(t)\geq 0$. 
        
        This schedule must execute $u_i$ before $i + 1$ days have passed or it will have urgency $u_i \leq i-(i+1)$ and hence $u_i<0$ before being executed. 
        
        Likewise for $u_{i+1}$ as $u_i = u_{i+1}$. As $U$ is ordered, all elements before $u_i$ are less than or equal to $u_i$ and hence less than or equal to $i$. Therefore they must also be executed before $i+1$ days have passed, by the same reasoning. Therefore at least $i+1$ tasks must be executed in the first $i$ days, which contradicts the rule that only one element may be performed daily.
        
        Therefore, $u_{i+1} \geq i+1$.
    \end{itemize}
    
    Therefore in all schedulable states $\forall i: u_i\geq i$.
\end{proofof}

\begin{proofof}{\wref{pro:forcing}}
    Let $U$ be an urgency state of a Pinwheel Scheduling system $A$, and let $i'$ be a task in this system such that $u_{i'}=i'$. Let $S$ be an arbitrary schedule for $A$ from $U$. As $U$ is ordered, $\forall i: u_i \leq u_{i+1}$, so for all elements preceding $i'$ it holds that $u_{i<i'}\leq u_{i'}$ and hence $u_{i<i'}\leq i$. Assume towards a contradiction that there exists a $i''\leq i'$ such that the task at $i''$ is executed after $i'$ days. Initially, $u_{i''}(0)\leq i$. $u_{i''}$ decreases by 1 each day, so if it is executed on day $t$, before it is executed $u_{i''}(t\geq i+1)\leq i - (i+1) < 0$. This contradicts the safety condition that $\forall i,t: u_i(t) \geq 0$.
    
    Therefore no task at positions $\leq i'$ can be executed later than day $i'$.
\end{proofof}

    \section{Worked Examples}
\label{app:worked-examples}

This appendix contains some examples to illustrate the impact of optimisations.
The core procedure is represented in \wref{fig:backtracking procedure}.

\begin{figure}[hbt]
    \centering
    \includegraphics[width=8cm]{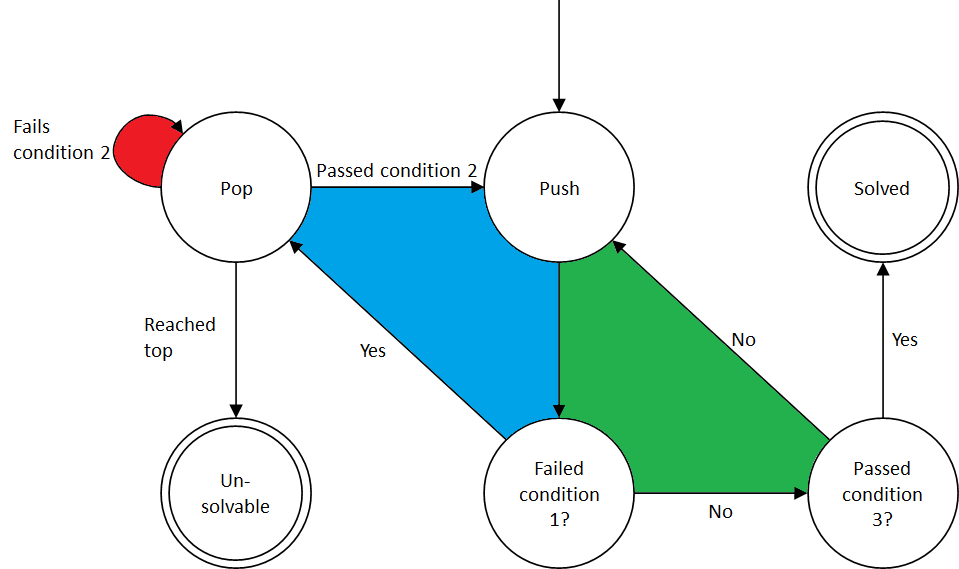}
    \caption{%
    	The backtracking procedure used to solve Pinwheel Scheduling problems. The left hand (red) loop climbs the tree when nodes have no potentially valid descendants, shortening $S_c$. The central (blue) loop shows letters being tried until one is found which is valid, moving horizontally. The right hand (green) loop shows the extension of the tree when the safety conditions are passed.}
    \label{fig:backtracking procedure}
\end{figure}

\textsl{\textbf{Note:} The examples follow the convention of the implementation,
which differs from the presentation in the main text by
(1) indexing tasks starting at $0$ (not $1$), and
(2) Pinwheel Scheduling instances are listed by weakly \emph{decreasing} frequency,
not increasing.}

The description below uses three conditions for a solution prefix $S$
and the corresponding state $X$:
\begin{enumerate}
    \item \(S\) is feasible (\(\forall t,i:x_i(t)\leq a_i\)).
    \item \(S\) is not known to fail in the future (it has unexamined extensions that may succeed).
    \item \(S\) eventually succeeds (\(X(p+t) = X(t) \) for some \(t\) such that \(p > 0\)).
\end{enumerate}
Any sequence which obeys the first two conditions (the safety conditions) is a \emph{candidate solution prefix} (hence \(S_c\)). To find out that and instance \(A\) is unsolvable, it must be shown that all candidate solution prefixes are eliminated by failing the safety conditions. The third condition invokes periodicity~-- any path that returns to somewhere it has been can be followed indefinitely, returning to the that location after each loop. Thus if \(S_c\) has this property it can be safely repeated indefinitely, which makes \(S_c\) a solution to \(A\).

\subsection{Worked Examples}

\subsubsection{Naïve}
\label{app:naive-worked-examples}

Examples of the trees generated by the Naïve method are shown in \wref{fig:Naive 6,3,2} (unschedulable) and \wref{fig:Naive 6,3,3} (schedulable). These are implemented serially, with the following preferences: If possible, push letter 0 onto the stack. If this is impossible, push the next letter, if there is one. If this is impossible pop the last letter. 

\textsl{\textbf{Note:} Solutions are not shown explicitly, but can be recovered for each path through the tree according to the placement of 0's}

\begin{figure}
    \centering
    \includegraphics[width=8cm]{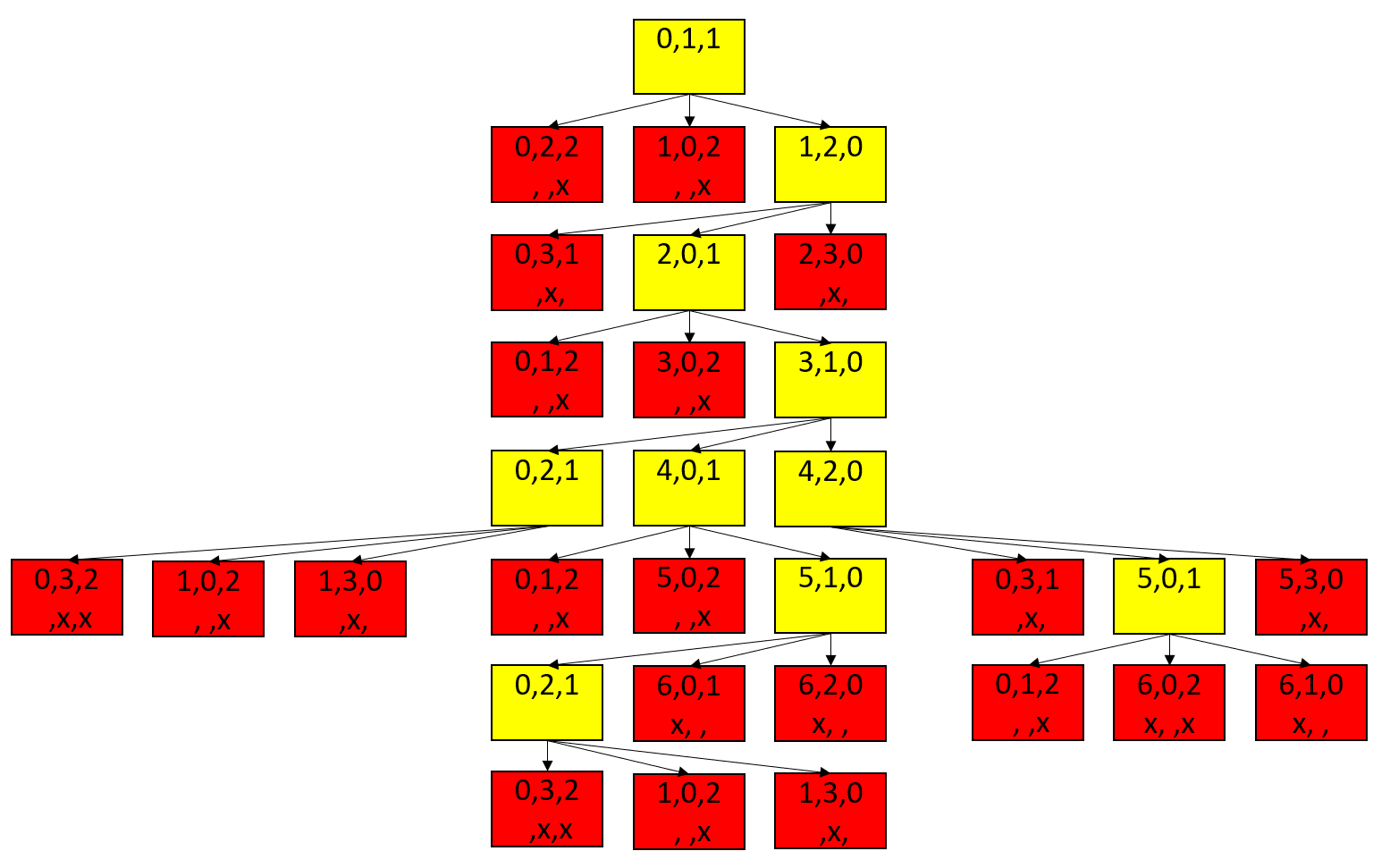}
    \caption{An example of the Naïve procedure from \wref{sec:Naive} for the input (6, 3, 2). States which pass condition 1 are shown in yellow, while states which fail it are shown in red. Every state with descendants executes tasks 1, 2, 3 in order from left to right. As the tree is closed, success testing does not find a repeated condition, but success testing would compare all yellow elements to all of their ancestor elements for a total cost of 34 comparisons.}
    \label{fig:Naive 6,3,2}
\end{figure}

\begin{figure}
    \centering
    \includegraphics[width=8cm]{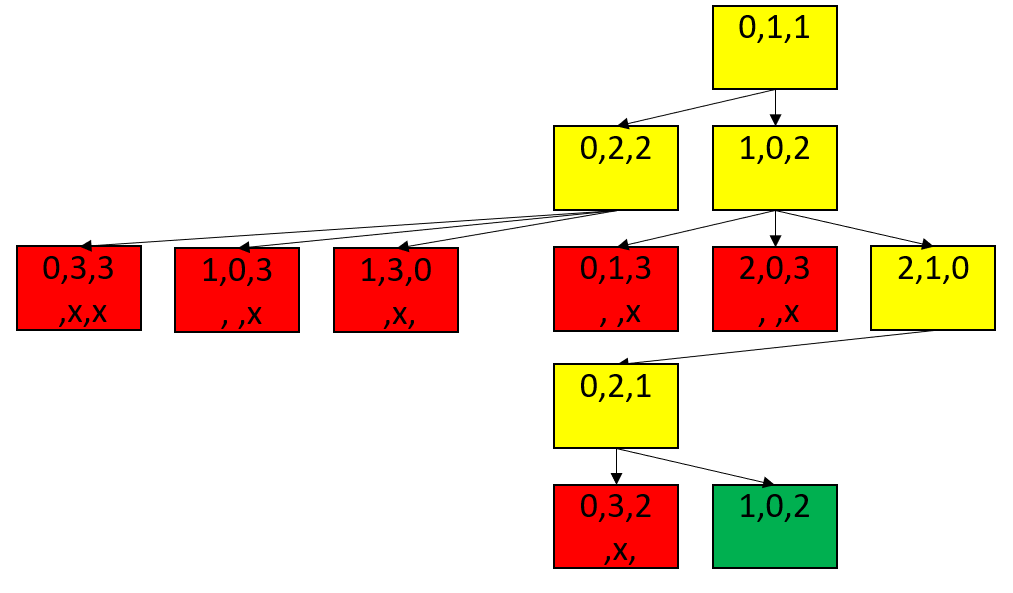}
    \caption{An example of the Naïve procedure from \wref{sec:Naive} for the input (6, 3, 3). States which pass condition 1 are shown in yellow, states which fail it are shown in red and the eventual successful case is shown in green. Every state with descendants executes tasks 1, 2, 3 in order from left to right, though several states have fewer than three descendants as the tree is open and so only partially explored. Every yellow state was compared to all of its ancestor elements while the green case was only compared until an identical state was found, for a total cost of 9 comparisons.
    If this example had used the most urgent to least urgent order described in \wref{sec:Naive} it would have examined 27 nodes, for a testing cost of 104 comparisons.}
    \label{fig:Naive 6,3,3}
\end{figure}

\subsubsection{Optimised}
\label{subsubsection: optimised worked examples}

Worked examples for the cases shown in \wref{app:naive-worked-examples} are repeated using the Optimised algorithm in \wref[Figures]{fig:optimised 6,3,2} and~\ref{fig:optimised 6,3,3}.

\begin{figure}
    \centering
    \includegraphics[width=7cm]{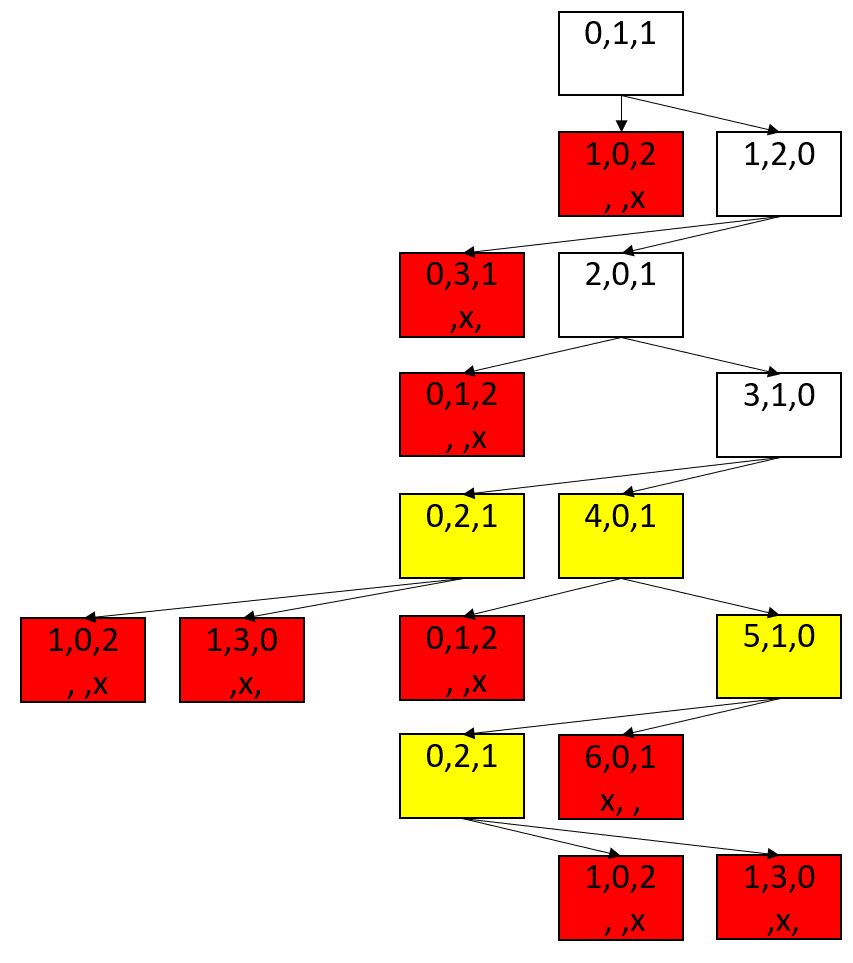}
    \caption{An example of the Optimised algorithm from \wref{sec:optimisations} for the input (6, 3, 2). States which pass condition 1 are shown in yellow or white, while states which fail it are shown in red. As opposed to the Naïve method for solving the same input (\wref{fig:Naive 6,3,2}), only tasks which were not performed the previous day are executed. The minimum solution length for this input is 4, so only nodes deeper than 4 (coloured yellow) require any success testing. Each such node is compared with all elements at least 4 days older than itself for a total of 7 comparisons and 17 nodes (while the Naïve algorithm uses 34 comparisons and 31 nodes for this instance).}
    \label{fig:optimised 6,3,2}
\end{figure}

\begin{figure}
    \centering
    \includegraphics[width=3.5cm]{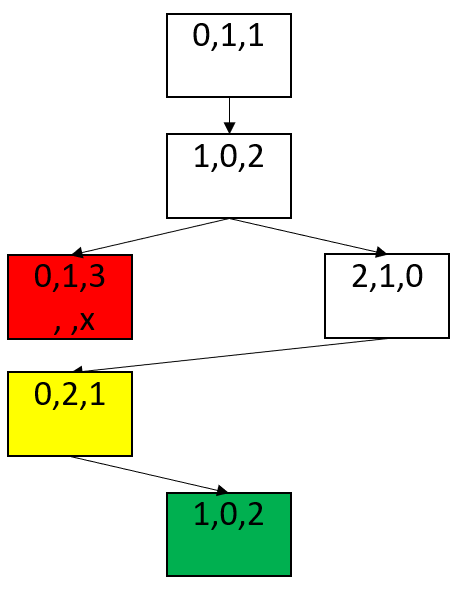}
    \caption{An example of the Optimised algorithm from \wref{sec:optimisations} for the input (6, 3, 3). States which pass condition 1 are shown in yellow, white or green; while states which fail it are shown in red. As opposed to the Naïve method for solving the same input (\wref{fig:Naive 6,3,3}), only tasks which were not performed the previous day are executed. The minimum solution length for this input is 3, so only nodes deeper than 3 (coloured yellow or green) require any success testing. Each such node is compared with all elements at least 3 days older than itself for a total of 3 comparisons and 6 nodes (while the Naïve algorithm uses 9 comparisons and 12 nodes for this instance).}
    \label{fig:optimised 6,3,3}
\end{figure}

}

\ifdraft{\ifthenelse{\boolean{collectnotes}}{
	\clearpage
	\part*{Notes-to-self}
	\printnotestoself
}{}}{}

\end{document}